\documentclass[journal,10pt]{IEEEtran}
\usepackage{blindtext}
\usepackage{graphicx}
\usepackage{amsmath}
\usepackage{mathrsfs}
\usepackage{amsthm}
\usepackage{multirow}
\usepackage{amsmath}
\usepackage{mathtools}

\DeclareMathOperator*{\argmin}{argmin}
\usepackage{yhmath}
\usepackage{amssymb}
\usepackage{array}
\usepackage{leftidx}
\usepackage{longtable}
\usepackage{stfloats}
\usepackage{cite}
\usepackage{url}
\usepackage{subfigure}
\usepackage{comment}
\usepackage{color}
\usepackage{algorithm}
\usepackage{algorithmic}
\usepackage{lipsum}



\begin{document}
\theoremstyle{plain}
\newtheorem{thm}{Theorem}
\newtheorem{remark}{Remark}
\newtheorem{lemma}{Lemma}
\newtheorem{prop}{Proposition}
\newtheorem*{cor}{Corollary}
\theoremstyle{definition}
\newtheorem{defn}{Definition}
\newtheorem{condi}{Condition}
\newtheorem{assump}{Assumption}

\title{Game Projection and Robustness for Game-Theoretic Autonomous Driving}
\author{Mushuang Liu\IEEEauthorrefmark{1},~\IEEEmembership{Member,~IEEE,}  ~H. Eric Tseng\IEEEauthorrefmark{3},  ~Dimitar Filev\IEEEauthorrefmark{3},~\IEEEmembership{Fellow,~IEEE,} ~ Anouck Girard\IEEEauthorrefmark{2},~\IEEEmembership{Senior Member,~IEEE} and Ilya Kolmanovsky\IEEEauthorrefmark{2},~\IEEEmembership{Fellow,~IEEE}

\thanks{\IEEEauthorrefmark{1} M. Liu is with the Department of Mechanical and Aerospace Engineering, University of Missouri, Columbia, MO, USA (email: ml529@missouri.edu).
}

\thanks{\IEEEauthorrefmark{3} H. E. Tseng, and D. Filev are with Ford Research and Innovation Center, 2101 Village Road, Dearborn, MI 48124, USA (e-mail: htseng@ford.com and dfilev@ford.com).}
\thanks{\IEEEauthorrefmark{2} A. Girard and I. Kolmanovsky are with the Department of  Aerospace Engineering, University of Michigan, Ann Arbor, MI, USA (email: anouck@umich.edu and ilya@umich.edu).
}
\thanks{This work is supported by Ford Motor Company.}
}
\maketitle
\markboth{Submitted to IEEE Transactions on Intelligent Transportation Systems} {Liu \MakeLowercase{\textit{et al.}}: }

\begin{abstract}
 Game-theoretic approaches are envisioned to bring human-like reasoning skills and decision-making processes for autonomous vehicles (AVs). However,  challenges including \textit{game complexity} and \textit{incomplete information} still remain to be addressed before they can be sufficiently practical for real-world use. \textit{Game complexity} refers to the difficulties of solving a multi-player game, which include solution existence, algorithm convergence, and scalability. To address these difficulties, a potential game based framework was developed in our recent work \cite{my_potential}. However, conditions on cost function design need to be enforced to make the game a potential game. This paper relaxes the conditions and makes the potential game approach applicable to more general scenarios, even including the ones that cannot be molded as a potential game. \textit{Incomplete information} refers to the ego vehicle's lack of knowledge of other traffic agents' cost functions. Cost function deviations between the ego vehicle estimated/learned other agents' cost functions and their actual ones are often inevitable. This motivates us to study the robustness of a game-theoretic solution. This paper defines the robustness margin of a game solution as the maximum magnitude of cost function deviations that can be accommodated in a game without changing the optimality of the game solution. With this definition, closed-form robustness margins are derived. Numerical studies using highway lane-changing scenarios are reported.   
\end{abstract}
\begin{IEEEkeywords}
Game theory, potential games, game decomposition, game projection, game robustness 
\end{IEEEkeywords}
\section{Introduction}
Fully autonomous vehicles (AVs) are envisioned to bring significant societal benefits, including improved road safety \cite{roadsafety}, increased fuel efficiency \cite{fuelefficiency}, reduced traffic congestion \cite{trafficcongestion}, enhanced mobility for the disabled and elderly \cite{inhancedmobility}, increased productivity \cite{productivity}, reduced need for parking space \cite{parking}, smoother traffic flow \cite{trafficflow}, and new business opportunities \cite{newbusiness}. However, major technical challenges still remain to be addressed before AVs can routinely drive on public roads, and one of them is the AV decision-making algorithm design \cite{challenge1_2022,challenge2_2022}. The AV decision-making algorithm serves as the AV's ``brain" and directly determines the AV behavior such as how the AV responds to the environment and/or to other road users, which may include AVs, human-driven vehicles, bicycles, and pedestrians. 

To enable interactive AVs with human-like reasoning skills and behavior, game-theoretic approaches have been explored in the literature \cite{game_1,game_merge,game_pursuer,game_racing,mine_1,Victor,online_payoff,game_leader}. In game theory, agents are assumed to pursue self-interests while considering the interactions with others. Such a reasoning process is often considered consistent with humans' (and human drivers') reasoning and decision making \cite{human_reasoning_1,pcpg_paper}. Along this line of research, both Nash games \cite{mine_1,Victor} and Stackelberg games (also called leader-follower games) \cite{online_payoff,game_leader} have been explored in the context of autonomous driving. However, these game-theoretic approaches often suffer from at least one of the two major challenges: \textit{Game complexity} and \textit{Incomplete information}. 

Game complexity refers to the difficulties of solving a multi-player game. These difficulties include, but are not limited to, (a) Solution existence, i.e., the desired solution, e.g., pure strategy Nash equilibrium (PSNE), may not always exist; (b) Algorithm convergence, i.e., the solution-seeking algorithm, e.g., best or better response dynamics algorithms, may not always converge; and (c) Scalability, i.e., the computational complexity of solving a multi-player game can increase exponentially with the number of players \cite{BR,game_leader}.  A few approaches/techniques have been developed to address some of these challenges. For example,  to address the scalability challenge, pairwise games have been widely used \cite{game_leader,online_payoff,Victor,mine_1}. However, such a pairwise setting may lead to conservative AV behavior and even deadlocks \cite{deadlock}, as it only captures local interactions between a pair of agents as opposed to global interaction among all agents. To overcome this limitation, as well as to address the aforementioned challenges (a)-(c), a receding horizon potential game based framework has been developed in our recent studies \cite{my_potential}. To construct a potential game, conditions on cost function design were established in \cite{my_potential}. However, although the proposed cost function design can accommodate many driving objectives, there are still traffic scenarios/situations where agents' behavior cannot be modeled by the proposed cost function \cite{game_leader,mine_1,suzhou}. Therefore, approaches to extending the applicability of the potential game approach to situations where the cost function may not satisfy the conditions in \cite{my_potential} or even the game is not potential are of significant interest. 

Incomplete information, which is another major challenge in game-theoretic autonomous driving, refers to the ego vehicle's lack of global knowledge of other traffic agents' objective functions or cost functions. Agents' cost functions should reflect their driving preferences/habits/styles, which are often encoded as driver-specific parameters in the cost functions \cite{online_payoff,closedloop1}. To know the values of these parameters, both offline calibration methods \cite{validation_1,validation_2,validation_3,validation_4,payoff1,closedloop1} and online learning and estimation methods have been developed \cite{online_payoff,sov}. However, irrespective of the approach, when facing specific agents on road, cost function deviations between the ego vehicle estimated/learned other agents' cost functions and their actual ones could always exist \cite{pcpg_paper}. Due to this inevitable cost function deviation, it is important that the obtained game-theoretic solution has a certain level of robustness, such that the solution of the assumed game (i.e., the game with the ego vehicle estimated/learned cost functions of others) can still perform well and possibly even remain optimal, in the actual game (i.e., the game of agents' actual cost functions). Therefore,  robustness analysis of game-theoretic autonomous driving policies needs to be conducted to find the robustness margins of a game-theoretic solution. 

Motivated by the above considerations, this paper focuses on the development of a more practical and robust game-theoretic framework for autonomous driving. Towards practicality, this paper broadens the applicability of the potential game framework for autonomous driving \cite{my_potential} to the case of non-potential games, by exploiting the game decomposition techniques \cite{game_decomposition,game_decomposition_2}. Specifically, given an arbitrary game, it is possible to project it onto the potential game space and approximate it using its ``closest" potential game. We provide such a game projection algorithm in this paper. Towards robustness, this paper defines and derives the robustness margin of a game-theoretic solution in the context of autonomous driving. The robustness margin is defined as the maximum magnitude of cost function deviations, which may be caused by the game projection and/or by the ego vehicle incomplete information, such that the obtained PSNE remains a PSNE in the original game, which corresponds to the game before projection and/or the game of agents with the actual cost functions. 

This paper is organized as follows. Section \ref{preliminary} introduces necessary definitions and preliminaries needed in the later sections. Section \ref{Sec_problemformulation} formulates the AV decision-making problem as a receding horizon multi-player game problem. Section \ref{sec_projection} introduces the game decomposition technique to project an arbitrary game onto the potential game space. Section \ref{sec_robustness} establishes robustness analysis. Section \ref{sec:simulation} includes a specific autonomous driving example as numerical studies, and Section \ref{sec_conclusion} concludes the paper.

\section{Definitions and Preliminaries}\label{preliminary}

Let $\mathbb{Z}_{+}(\mathbb{Z}_{++})$  denote the set of non-negative (positive) integers and $\mathbb{R}_{+}(\mathbb{R}_{++})$  denote the set of non-negative (positive) real numbers. 

To introduce game-theoretic background, let us first define \textit{strategic-form finite games}.
\begin{defn}[Strategic-Form Finite Game \cite{game_book}]
    A strategic-form finite game is defined as a tuple $\mathcal{G}=\{\mathcal{N},\mathcal{A},\{J_i\}_{i\in\mathcal{N}}\}$, where
    \begin{itemize}
        \item $\mathcal{N}=\{1,2,\cdots,N\}$ is a finite set of players or agents.
        \item $\mathcal{A}=\{\mathcal{A}_1,\mathcal{A}_2,\cdots,\mathcal{A}_N\}$ is a set of strategies with $\mathcal{A}_i$ representing the finite strategy space of player $i$. 
        \item $J_i: \mathcal{A}\to \mathbb{R}$ is the cost function of player $i$.
    \end{itemize}
\end{defn}

 Denote {by} $\mathcal{N}_{-i}$ the set of all agents except for agent $i$, i.e., $\mathcal{N}_{-i}=\{1,2,\cdots,i-1,i+1,\cdots,N\}$. We let $\mathbf{a}_i\in\mathcal{A}_i$ represent the strategy of agent $i$ and $\mathbf{a}_{-i}\in\mathcal{A}_{-i}$ represent the set of strategies of all other agents except for agent $i$, i.e., $\mathbf{a}_{-i}=\{\mathbf{a}_1,\dots,\mathbf{a}_{i-1},\mathbf{a}_{i+1},\dots,\mathbf{a}_N\}$ with $\mathcal{A}_{-i}$ being the  domain of $\mathbf{a}_{-i}$. Denote $\mathbf{a}=\{\mathbf{a}_i,\mathbf{a}_{-i}\}\in\mathcal{A}$, which shall be referred to as a \textit{strategy profile}. Denote by $h_i=|\mathcal{A}_i|$ the cardinality of the strategy space of player $i$, and $|\mathcal{A}|=\prod_{i=1}^Nh_i$.
Let $J_{-i}=\{J_1,\dots,J_{i-1},J_{i+1},\dots,J_N\}$ and $J=\{J_i,J_{-i}\}$. 
 
 We denote by $C_0=\{f|f:\mathcal{A}\rightarrow \mathbb{R}\}$ the set of real-valued functions. The  inner product in $C_0$ is defined as \cite{game_decomposition}.
 \begin{equation}\label{inner_1}
     \langle f_1,f_2\rangle=\sum_{\mathbf{a}\in\mathcal{A}}f_1(\mathbf{a})f_2(\mathbf{a}).
 \end{equation}
 
 Given  $\mathcal{N}$ and $\mathcal{A}$, a game is uniquely defined by its set of cost functions $\{J_i\}_{i\in\mathcal{N}}\in C_0^N$. Hence, the \textit{space of games} with players $\mathcal{N}$ and strategy profiles $\mathcal{A}$, which is denoted by $\mathscr{G}_{\mathcal{N},\mathcal{A}}$, can be identified as $\mathscr{G}_{\mathcal{N},\mathcal{A}}\cong C_0^N$. 
Define the inner product on the space of games $\mathscr{G}_{\mathcal{N},\mathcal{A}}$ \cite{game_decomposition} as  \begin{equation}\label{game_inner}
\langle\mathcal{G},\hat{\mathcal{G}}\rangle_{\mathcal{N},\mathcal{A}}\triangleq \sum_{i\in \mathcal{N}}h_i\langle J_i,\hat{J}_i\rangle,
 \end{equation}
 where $\mathcal{G}=\{\mathcal{N},\mathcal{A},\{J_i\}_{i\in\mathcal{N}}\}$ and $\hat{\mathcal{G}}=\{\mathcal{N},\mathcal{A},\{\hat{J}_i\}_{i\in\mathcal{N}}\}$.
 
 The inner product \eqref{game_inner} induces a norm to quantify the distance between games. Specifically, the norm on $\mathscr{G}_{\mathcal{N},\mathcal{A}}$ \cite{game_decomposition} is defined as 
 \begin{equation}\label{game_norm}
     \|\mathcal{G}\|^2_{\mathcal{N},\mathcal{A}}=\langle\mathcal{G},\mathcal{G}\rangle_{\mathcal{N}, \mathcal{A}}.
 \end{equation}

Several basic solution concepts for a non-cooperative game, including the \textit{best response}, \textit{Nash equilibrium}, and \textit{$\epsilon$-Nash equilibrium} are introduced in Definitions \ref{d1}-\ref{d_epsilon}.
\begin{defn} [Best Response \cite{game_book}] \label{d1} 
Agent $i$'s best response to other agents' fixed strategies $\mathbf{a}_{-i}\in\mathcal{A}_{-i}$ is defined as the strategy $\mathbf{a}_i^*$ such that
\begin{equation}
    J_i( \mathbf{a}_i^*, \mathbf{a}_{-i})\leq J_i(\mathbf{a}_i, \mathbf{a}_{-i}) \quad \forall \mathbf{a}_i\in\mathcal{A}_i.
\end{equation}
\end{defn}

\begin{defn} [Pure-Strategy Nash Equilibrium \cite{game_book}] \label{d2}
An $N$-tuple of strategies (or strategy profile) $\{\mathbf{a}_1^*,\mathbf{a}_2^*,...,\mathbf{a}_N^*\}$  is  a pure-strategy Nash equilibrium for an $N$-player game if and only if   
\begin{equation}\label{Nash_de}
 J_i(\mathbf{a}_i^*,\mathbf{a}_{-i}^*)\leq J_i(\mathbf{a}_i,\mathbf{a}_{-i}^*)\quad  \forall \mathbf{a}_i\in\mathcal{A}_i, \forall i\in\mathcal{N}.
\end{equation}
\end{defn}
Equation \eqref{Nash_de} implies that  if all agents play their best response, then a PSNE is achieved, and if a PSNE is achieved, then no player would have the incentive to change its strategy. Throughout this paper, we only consider pure strategies, and the term ``pure strategy" may be omitted.

Related to NE, we define \textit{$\epsilon$-NE} as follows. 
\begin{defn} [$\epsilon$ Nash Equilibrium \cite{game_decomposition}] \label{d_epsilon}
Given $\epsilon\in \mathbb{R}_{+}$, a strategy profile $\{\mathbf{a}_1,\mathbf{a}_2,...,\mathbf{a}_N\}$  is  a pure-strategy $\epsilon$-NE for an $N$-player game if and only if   
\begin{equation}\label{epsilonNash_de}
 J_i(\mathbf{a}_i,\mathbf{a}_{-i})\leq J_i(\mathbf{a}'_i,\mathbf{a}_{-i})+\epsilon\quad  \forall \mathbf{a}'_i\in\mathcal{A}_i, \forall i\in\mathcal{N}.
\end{equation}
\end{defn}

As shown by \eqref{Nash_de} and \eqref{epsilonNash_de}, both NE and $\epsilon$-NE are defined in terms of cost differences at \textit{comparable strategy profiles}, which are defined as a pair of strategy profiles that only differ in a single agent's strategy \cite{game_decomposition}. We denote the set of pairs of comparable strategy profiles by $\mathcal{C}\subset\mathcal{A}\times\mathcal{A}$, i.e., the strategy profiles $\mathbf{a}\in\mathcal{A}$ and $\mathbf{b}\in\mathcal{A}$ are comparable if $(\mathbf{a},\mathbf{b})\in\mathcal{C}$. A pair of comparable strategy profiles that differ in agent $i$'s strategy is referred to as \textit{$i$-comparable strategy profiles}. We denote the set of pairs of $i$-comparable strategy profiles by $\mathcal{C}_i\subset\mathcal{A}\times\mathcal{A}$, i.e., the strategy profiles $\mathbf{a}\in\mathcal{A}$ and $\mathbf{b}\in\mathcal{A}$ are $i$-comparable if $(\mathbf{a},\mathbf{b})\in\mathcal{C}_i$. With this, we define a \textit{pairwise comparison function} $X:\mathcal{A}\times\mathcal{A}\rightarrow \mathbb{R}$ as 
\begin{align}\label{X}
&X(\mathbf{a},\mathbf{b})\nonumber\\
&=\begin{cases}J_i(\mathbf{a})-J_i(\mathbf{b}), &\mbox{if $(\mathbf{a},\mathbf{b})\in\mathcal{C}_i$ for some $i\in\mathcal{N}$},\vspace{1ex}\\
0, &\mbox{otherwise}.\vspace{1ex}\end{cases}
\end{align}

Note that $X(\mathbf{a},\mathbf{b})$ satisfies 
\begin{align}\label{X_property}
&X(\mathbf{a},\mathbf{b})=\begin{cases} -X(\mathbf{b},\mathbf{a}), &\mbox{if $(\mathbf{a},\mathbf{b})\in\mathcal{C}$},\vspace{1ex}\\
0, &\mbox{otherwise}.\vspace{1ex}\end{cases}
\end{align}

Such a pairwise comparison function results in a \textit{flow representation} of a game \cite{game_decomposition}, which emphasizes the change of costs between comparable strategy profiles. 

We denote by $C_1=\{Y|Y:\mathcal{A}\times\mathcal{A}\rightarrow \mathbb{R}\}$ the set of functions that satisfy \eqref{X_property}. In addition, we define an indicator function $W:\mathcal{A}\times\mathcal{A}\rightarrow \mathbb{R}$ as
\begin{align}\label{W}
&W(\mathbf{a},\mathbf{b})=\begin{cases} 1, &\mbox{if $(\mathbf{a},\mathbf{b})\in\mathcal{C}$},\vspace{1ex}\\
0, &\mbox{otherwise}.\vspace{1ex}\end{cases}
\end{align}

With the above notation, we define the \textit{combinatorial gradient operator} $\delta_0:C_0\rightarrow C_1$, given by 
\begin{equation}\label{delta_def}
    (\delta_0f)(\mathbf{a},\mathbf{b})=W(\mathbf{a},\mathbf{b})(f(\mathbf{a})-f(\mathbf{b})),
\end{equation}
where $\mathbf{a},\mathbf{b}\in \mathcal{A}$ and  $f\in C_0$ \cite{game_decomposition}.

Now we introduce player-specific operators. Let $W_i:\mathcal{A}\times\mathcal{A}\rightarrow \mathbb{R}$ be the indicator function for $i$-comparable strategy profiles given by
\begin{align}\label{Wi}
&W_i(\mathbf{a},\mathbf{b})=\begin{cases} 1, &\mbox{if $(\mathbf{a},\mathbf{b})\in\mathcal{C}_i$},\vspace{1ex}\\
0, &\mbox{otherwise}.\vspace{1ex}\end{cases}
\end{align}
It relates to $W$ in \eqref{W} by
\begin{equation}
    W=\sum_{i\in\mathcal{N}}W_i.
\end{equation}

Define the difference operator $D_i:C_0\rightarrow C_1$ such that
\begin{equation}\label{operator_difference}
    (D_if)(\mathbf{a},\mathbf{b})=W_i(\mathbf{a},\mathbf{b})(f(\mathbf{a})-f(\mathbf{b})),
\end{equation}
where $\mathbf{a},\mathbf{b}\in \mathcal{A}$ and  $f\in C_0$ \cite{game_decomposition}. This operator quantifies the change in $f$ between a pair of $i$-comparable strategy profiles. With this, the pairwise comparison function \eqref{X} can be represented as 
\begin{equation}\label{D_def}
    X=\sum_{i\in\mathcal{N}}D_iJ_i.
\end{equation}

Relevant to $D_i$, we define $D:C_0^N\rightarrow C_1$ such that $D=[D_1,D_2,..,D_N]$.

Next we consider three classes of games, called \textit{finite potential games}, \textit{finite harmonic games}, and \textit{finite nonstrategic games}. Throughout this paper, we only consider finite games, i.e., the games with a finite number of players and strategies, and the term ``finite" may be omitted. 

\begin{defn} [Finite Potential Game \cite{potential_book}] \label{d3}
The game $\mathcal{G}=\{\mathcal{N},\mathcal{A},\{J_i\}_{i\in\mathcal{N}}\}$ is a finite potential game if and only if $\mathcal{N}$ and $\mathcal{A}$ contain a finite number of elements  and a potential function {$F:\mathcal{A}\rightarrow \mathbb{R}$} exists such that, $\forall i\in\mathcal{N}$,
\begin{equation}\label{Potential_def}
\begin{split}
& J_i(\mathbf{a}_i,\mathbf{a}_{-i})- J_i(\mathbf{a}'_i,\mathbf{a}_{-i})= F(\mathbf{a}_i,\mathbf{a}_{-i})- F(\mathbf{a}'_i,\mathbf{a}_{-i}),\\
&\quad\quad\quad \forall \mathbf{a}_i,\mathbf{a}'_i\in \mathcal{A}_i, \text{and }   \forall \mathbf{a}_{-i}\in \mathcal{A}_{-i}.
\end{split}
\end{equation}
\end{defn}

With the difference operator $D_i$ defined in \eqref{operator_difference}, the potential game can be alternatively defined as the game where a function $F:\mathcal{A}\rightarrow \mathbb{R}$ exists such that 
\begin{equation}\label{potential_def_alter}
    D_iJ_i=D_iF
\end{equation} 
holds for all $i\in\mathcal{N}$ \cite{potential_projection}. 
A potential game has many appealing properties. One of them is that a potential game always has at least one PSNE \cite{potential_book, my_potential}. 

 \begin{defn} [Finite Harmonic Game \cite{game_decomposition}] \label{d_h}
The game $\mathcal{G}=\{\mathcal{N},\mathcal{A},\{J_i\}_{i\in\mathcal{N}}\}$ is a finite harmonic game if and only if $\mathcal{N}$ and $\mathcal{A}$  contain a finite number of elements and
\begin{equation}\label{harmonic}
\begin{split}
& \sum_{i\in\mathcal{N}}\sum_{\mathbf{a}'_i\in\mathcal{A}_i}\left(J_i(\mathbf{a}_i,\mathbf{a}_{-i})- J_i(\mathbf{a}'_i,\mathbf{a}_{-i})\right)=0,\\
&\quad\quad\quad \forall \mathbf{a}_i,\mathbf{a}'_i\in \mathcal{A}_i, \text{and }   \forall \mathbf{a}_{-i}\in \mathcal{A}_{-i}.
\end{split}
\end{equation}
\end{defn}

Eq. \eqref{harmonic} indicates that harmonic games are games that satisfy flow conservation condition: For any strategy profile $(\mathbf{a}_i,\mathbf{a}_{-i})$, its total ``flow", i.e., the sum of all pairwise comparisons   $\sum_{i\in\mathcal{N}}\sum_{\mathbf{a}'_i\in\mathcal{A}_i}\left(J_i(\mathbf{a}_i,\mathbf{a}_{-i})- J_i(\mathbf{a}'_i,\mathbf{a}_{-i})\right)$, equals zero.   
In contrast to potential games, a harmonic game does not have a PSNE, unless it is a nonstrategic game as defined in Definition \ref{d_N}.

 \begin{defn} [Finite Nonstrategic Game \cite{nonstrategic}] \label{d_N}
The game $\mathcal{G}=\{\mathcal{N},\mathcal{A},\{J_i\}_{i\in\mathcal{N}}\}$ is a finite nonstrategic game if and only if $\mathcal{N}$ and $\mathcal{A}$  contain a finite number of elements and $\forall i\in\mathcal{N}$,
\begin{equation}\label{nonstrategic}
\begin{split}
& J_i(\mathbf{a}_i,\mathbf{a}_{-i})= J_i(\mathbf{a}'_i,\mathbf{a}_{-i}),\\
&\forall \mathbf{a}_i,\mathbf{a}'_i\in \mathcal{A}_i, \text{and }   \forall \mathbf{a}_{-i}\in \mathcal{A}_{-i}.
\end{split}
\end{equation}
\end{defn}

Eq. \eqref{nonstrategic} shows that for a nonstrategic game, every strategy profile is a PSNE. In addition, given two games $\mathcal{G}=\{\mathcal{N},\mathcal{A},\{J_i\}_{i\in\mathcal{N}}\}$ and $\hat{\mathcal{G}}=\{\mathcal{N},\mathcal{A},\{\hat{J}_i\}_{i\in\mathcal{N}}\}$, they are strategically equivalent if the game $\{\mathcal{N},\mathcal{A},\{J_i-\hat{J}_i\}_{i\in\mathcal{N}}\}$ is a nonstrategic game \cite{game_decomposition_2}. Strategically equivalent games yield the same set of PSNE.

To fix a representative for strategically equivalent games \cite{game_decomposition},  we introduce a \textit{normalized game} as follows.
 \begin{defn} [Finite Normalized Game \cite{game_decomposition}] \label{d_Normal}
A game $\mathcal{G}=\{\mathcal{N},\mathcal{A},\{J_i\}_{i\in\mathcal{N}}\}$ is called normalized if 
\begin{equation}\label{normalized}
\begin{split}
& \sum_{\mathbf{a}_i\in \mathcal{A}_i}J_i(\mathbf{a}_i,\mathbf{a}_{-i})=0,\\
&\forall \mathbf{a}_{-i}\in \mathcal{A}_{-i}, \text{and }  \forall i\in \mathcal{N}.
\end{split}
\end{equation}
\end{defn}

As shown in \cite{game_decomposition}, for any given game, there exists a unique strategically equivalent normalized game, i.e., these two games have identical pairwise comparison functions.

\section{Problem Formulation}\label{Sec_problemformulation}
Consider a traffic scenario, where multiple traffic agents  $\mathcal{N}=\{1,2,\cdots,N\}$ share the road. Each agent aims to optimize its own driving performance, which is jointly affected by other road users so that the optimal strategies are determined according to   
\begin{equation}\label{optimization}
    \begin{split}
    \mathbf{a}^*_i(t)&\in \argmin_{\mathbf{a}_i(t)\in\mathcal{A}_i}J_i(\mathbf{a}_i(t),\mathbf{a}_{-i}(t))
    \\
    &=\argmin_{\mathbf{a}_i(t)\in\mathcal{A}_i} \sum_{\tau=t}^{t+T-1}\Psi_i(x_i(\tau),x_{-i}(\tau),a_i(\tau),a_{-i}(\tau)),
    \end{split}
\end{equation}
where  $t\in\mathbb{Z}_{+}$ is the time index, $\mathbf{a}_i(t)=\{a_i(t),a_i(t+1),\dots,a_i(t+T-1)\}\in\mathcal{A}_i$ is agent $i$'s strategy, i.e., agent $i$'s action sequence over the horizon $[t,t+1,...,t+T-1]$, $T\in\mathbb{Z}_{++}$ is the prediction horizon length, $a_i(t)\in\mathcal{U}_i$ is agent $i$'s action at $t$, which can represent, for instance,  a traffic agent's accelerations and/or steering angle, $x_i(t)\in\mathcal{X}_i$ is agent $i$'s state, which typically includes the position and velocity information, $\Psi_i:\mathcal{A}\times\mathcal{X}\rightarrow\mathbb{R}$, where $\mathcal{X}=\{\mathcal{X}_1,\mathcal{X}_2,...,\mathcal{X}_N\}$, is agent $i$'s instantaneous cost at one time instant, and $J_i:\mathcal{A}\rightarrow\mathbb{R}$ is agent $i$'s cumulative cost over the $T$ steps horizon. It is important to note that the expression of $J_i(\mathbf{a}_i(t),\mathbf{a}_{-i}(t))$ can be different at different $t$ due to its state dependence, i.e., dependence on states $x_i(t)$ and $x_{-i}(t)$. With a slight abuse of notation, we do not distinguish $J_i$ at different $t$, as this state dependence is clear from the context. After deriving $\mathbf{a}^*_i(t)$, agent $i$ implements the first element
 $a^*_i(t)$ and repeats the same procedure at the next time instant, $t+1$, with a shifted horizon.
 
 We use  a general nonlinear function $g_i$ to represent agent $i$'s state evolution model:
 \begin{equation}
     x_i(t+1)=g_i(x_i(t),a_i(t)).
 \end{equation}
Examples of $g_i$ include a single-mass model \cite{add_dynamics_RL,add_dynamics_nan}, a unicycle model \cite{add_dynamics_RTD}, and a bicycle model \cite{bicycle}. We refer the readers to our recent work \cite{pcpg_paper} for more detailed discussions on the physical interpretations of $x_i$, $a_i$, $g_i$, and $J_i$, possible models of $g_i$, and $J_i$, and the rationality of the receding horizon game \eqref{optimization} in the autonomous driving decision-making setting. 

 As described in \cite{pcpg_paper}, such a receding horizon game enables predictive capability for the AV, allows fast response to any changes in the environment, and tackles agents' interactions from a self-interest optimization perspective, which is consistent with human driving experience. However, solving such a multi-player game \eqref{optimization} is not straightforward, and the challenges include 
\begin{enumerate}
    \item \textbf{\textit{Existence of solution:}} The strategy profile $\{\mathbf{a}_1^*(t),\mathbf{a}_2^*(t),...,\mathbf{a}_N^*(t)\}$ that satisfies   \eqref{optimization} for all $i\in\mathcal{N}$ may not always exist.
    \item \textbf{\textit{Convergence of algorithm:}} A solution seeking algorithm, e.g., best- or better- response dynamics algorithm, may not always converge.
    \item \textbf{\textit{Nonuniqueness of PSNE:}} The game \eqref{optimization} may yield multiple solutions, i.e., multiple PSNE.  
    {Selecting the one that is preferable to others is often not straightforward.}  
    \item \textbf{\textit{Computational scalability:}} Solving  \eqref{optimization} is often computationally expensive when $N$ is large. Specifically, if using best response dynamics (i.e., Algorithm \ref{A_br}), then the number of required optimizations  can increase exponentially with the number of agents \cite{BR}. 
    \item \textbf{\textit{Lack of information:}} To solve \eqref{optimization}, the ego vehicle needs to know everyone's cost function, which is not realistic in a traffic setting, considering the variability of human driving styles.  
\end{enumerate}
\begin{algorithm}[t]
\caption{Best response dynamics to solve \eqref{optimization}} \label{A_br}
\hspace*{0.0in} {\bf Inputs:} \\ 
\hspace*{0.2in}Agent set  $\mathcal{N}$;  \\
\hspace*{0.2in}System state  $x(t)$;  \\
\hspace*{0.2in}System dynamics  $g_i$, $i\in\mathcal{N}$;\\
\hspace*{0.2in}Strategy space  $\mathcal{A}$;  \\
\hspace*{0.2in}Cost functions  $J_i$, $i\in\mathcal{N}$;\\
\hspace*{0.02in} {\bf Output:} \\
\hspace*{0.2in} PSNE $\mathbf{a}^*(t)$.\\
\hspace*{0.02in} {\bf Procedures:} 
\begin{algorithmic}[1]
\STATE Set \textit{NashCondition=False}
\STATE {\bf While} \textit{NashCondition=False} {\bf do}
\STATE \hspace*{0.2in}{\bf For} $i=1,2,...,N$ {\bf do}
\STATE \hspace*{0.3in} Find $\mathbf{a}_i^*(t)$ according to 
\hspace*{0.3in} 
\begin{equation}\nonumber
    \quad\quad\quad\mathbf{a}_i^*(t)\in\argmin_{\mathbf{a}_i(t)\in\mathcal{A}_i} J_i(\mathbf{a}_i(t),\mathbf{a}_{-i}(t)).
\end{equation}
\STATE \hspace*{0.3in} Update $\mathbf{a}_i(t)$ using $\mathbf{a}^*_i(t)$.
\STATE \hspace*{0.2in} {\bf End for}
\STATE \hspace*{0.2in}{\bf If} \\
\STATE  \hspace*{0.3in} $
    \mathbf{a}^*_i(t)\in\argmin\limits_{\mathbf{a}_i(t)\in\mathcal{A}_i} J_i(\mathbf{a}_i(t),\mathbf{a}_{-i}^*(t))$ holds    $\forall i\in\mathcal{N}$,
\STATE  \hspace*{0.16in} {\bf Then} \\
\STATE \hspace*{0.3in} Set \textit{NashCondition=True}.
\STATE \hspace*{0.2in}{\bf End if}
\STATE {\bf End while}
\end{algorithmic}
\end{algorithm}

To solve the challenges $1-4$, a receding horizon potential game framework was developed in our recent studies \cite{my_potential}. Specifically, as detailed in \cite{my_potential}, if the game \eqref{optimization} is a potential game, then the following  properties hold:
\begin{itemize}
\item The game \eqref{optimization} always has at least one PSNE.
\item Algorithm \ref{A_br} always converges.
\item One can find a PSNE that is both individually optimal (in the sense of self-interest optimization, i.e., NE) and socially optimal (in the sense of social interest optimization, where the ``social interest" is characterized by a potential function). Such a PSNE is often preferable to other solutions because it represents a considerate AV which cares about not only itself but also other traffic agents.
\item The computational challenge is addressed by the potential function optimization algorithm (Algorithm \ref{A2}), since this algorithm only requires \textit{one} optimization procedure regardless of the number of agents $N$.
\end{itemize}  

Despite the aforementioned merits, it is often not straightforward to make the game  \eqref{optimization} a potential game.  We show in \cite{my_potential} that if agents' cost functions are designed so that they satisfy certain conditions, then the resulting game is a potential game. We  include this result in Theorem \ref{t1}. 

\begin{thm}[Theorem 6 in \cite{my_potential}]\label{t1}
If the cost function $J_i$ in \eqref{optimization} satisfies
\begin{equation}\label{cost_form}
\begin{split}
    &J_i(\mathbf{a}_i(t),\mathbf{a}_{-i}(t))\\
    &=\alpha J_i^{self}(\mathbf{a}_i(t))+\beta \sum_{j\in\mathcal{N},j\neq i} J_{ij}(\mathbf{a}_i(t),\mathbf{a}_{j}(t)),
\end{split}
\end{equation}
where  $J_i^{self}: \mathcal{A}_i\to \mathbb{R}$ is a function {determined} solely by agent $i$'s strategy, $J_{ij}: \mathcal{A}_i\times\mathcal{A}_j\to \mathbb{R}$ satisfies
\begin{equation}\label{symmetric}
\begin{split}
  &\quad  J_{ij}(\mathbf{a}_i(t),\mathbf{a}_{j}(t))=J_{ji}(\mathbf{a}_j(t),\mathbf{a}_{i}(t)), \\
  &\forall i,j\in\mathcal{N}, i\neq j, \text{ and } \forall \mathbf{a}_i\in \mathcal{A}_i, \mathbf{a}_{j}\in \mathcal{A}_j,
    \end{split}
\end{equation}
and $\alpha$ and $\beta$ are two {real numbers}. Then the game \eqref{optimization} is a potential game with the following potential function, 
\begin{equation}\label{potential_design}
\begin{split}
        & F(\mathbf{a}(t))\\
    & =\alpha\sum_{i\in\mathcal{N}}J_i^{self}(\mathbf{a}_i(t))+\beta\sum_{i\in \mathcal{N}}\sum_{j\in\mathcal{N},j< i} J_{ij}(\mathbf{a}_i(t),\mathbf{a}_{j}(t)).
\end{split}
\end{equation}
\end{thm}

Theorem \ref{t1} states that if $J_i$ in \eqref{optimization} satisfies \eqref{cost_form}, then the game \eqref{optimization} is a  potential game. For the cost \eqref{cost_form}, the first term $J_i^{self}$ can be used to model self-centered objectives, including tracking a desired speed, optimizing ride comfort, maximizing fuel efficiency, and so on. The second term $J_{ij}$ can be used to model symmetric pairwise interactions, such as a pairwise collision penalty \cite{my_potential}.  

Although such a cost function design satisfying \eqref{cost_form}-\eqref{symmetric} can accommodate many driving objectives, there are still traffic scenarios where agents' interactions are asymmetric. For instance, this asymmetry may be caused by right-of-way/road priority rules or by significant vehicle size differences such as the interaction between a sedan and a huge truck in highway merging. The modeling of these asymmetric interactions could then result in $J_{ij}$ not satisfying \eqref{symmetric}, and thus invalidate the conditions of Theorem \ref{t1}, thereby preventing the  potential game approach, which has many appealing properties as detailed above, from being applicable to these scenarios/situations. 

In addition, as described in Challenge $5$, the ego vehicle often does not have complete information on other traffic agents' cost functions, but can only rely on assumed costs obtained either from offline calibration \cite{payoff1,closedloop1} or  online estimation/learning \cite{online_payoff}. Therefore, it is very likely that cost function deviations exist between the game solved by the ego vehicle and the one of agents' actual cost functions. Quantification of whether and how such cost function deviations affect the game solution is desired. 

Motivated by the above considerations, in this paper, we provide a method to  broaden the applicability of the potential game based autonomous driving \cite{my_potential} to accommodate more general traffic scenarios. In addition, we quantify the robustness margins of solutions to the game  \eqref{optimization} against  cost function deviations. 
\begin{algorithm}[t]
\caption{Potential function optimization to solve \eqref{optimization}} \label{A2}
\hspace*{0.0in} {\bf Input:} \\ 
\hspace*{0.2in}Agent set  $\mathcal{N}$;  \\
\hspace*{0.2in}System state  $x(t)$;  \\
\hspace*{0.2in}System dynamics  $g_i$, $i\in\mathcal{N}$;\\
\hspace*{0.2in}Strategy space  $\mathcal{A}$;  \\
\hspace*{0.2in}Cost functions  $J_i$, $i\in\mathcal{N}$, that satisfy \eqref{cost_form};\\
\hspace*{0.02in} {\bf Output:} \\
\hspace*{0.2in} PSNE $\mathbf{a}^*(t)$.\\
\hspace*{0.02in} {\bf Procedures:} 
\begin{algorithmic}[1]
\STATE Find  $F(\mathbf{a}(t)) $ according to \eqref{potential_design}. 
\STATE Find $\mathbf{a}^*(t)$ such that
\begin{equation}\label{potential_optimize}
    \mathbf{a}^*(t)=\argmin_{\mathbf{a}(t)\in\mathcal{A}} F(\mathbf{a}(t)).
\end{equation}
\end{algorithmic}
\end{algorithm}
\section{Game Projection}\label{sec_projection}
To expand the use of the potential game framework to address games where agents' cost functions may not satisfy \eqref{cost_form}-\eqref{symmetric} or which could be non-potential to begin with, we propose to use the game decomposition technique. This technique was developed based on Helmholtz decomposition   and on the flow representation of a game \cite{game_decomposition}. The main result is included in Lemma \ref{space_decomposition}.
\begin{lemma}[Theorem 4.1 in \cite{game_decomposition}]\label{space_decomposition}
   The space of games $\mathscr{G}_{\mathcal{N},\mathcal{A}}$ is a direct sum of the potential, harmonic, and nonstrategic subspaces, i.e., 
   \begin{equation}
       \mathscr{G}_{\mathcal{N},\mathcal{A}}=\mathscr{P}\oplus\mathscr{H}\oplus\mathscr{N},
   \end{equation}
   where the potential subspace $\mathscr{P}$ is the set of games that satisfy both \eqref{Potential_def} and \eqref{normalized}, the harmonic subspace $\mathscr{H}$ is the set of games that satisfy both \eqref{harmonic} and \eqref{normalized}, and nonstrategic subspace $\mathscr{N}$ is the set of games that satisfy \eqref{nonstrategic} .

   In particular, given a game $\mathcal{G}=\{\mathcal{N},\mathcal{A},\{J_i\}_{i\in\mathcal{N}}\}$, it can be uniquely decomposed into three components:
   \begin{itemize}
       \item Potential component: $J^P\triangleq D^{\dagger}\delta_0\delta_0^{\dagger}DJ$,
        \item Harmonic component: $J^H\triangleq D^{\dagger}(I-\delta_0\delta_0^{\dagger})DJ$,
             \item Nonstrategic component: $J^N\triangleq (I-D^{\dagger}D)J$,
   \end{itemize}
   where $J=\{J_i\}_{i\in\mathcal{N}}$, $J^P+J^H+J^N=J$, $\delta_0^{\dagger}$ is the (Moore-Penrose) pseudoinverse of the operator $\delta_0$ (defined in \eqref{delta_def}) with respect to the inner product in \eqref{inner_1}, and $D^{\dagger}$ is the pseudoinverse of the operator $D$ (defined under \eqref{D_def}) with respect to the inner product in $C_0^N$, which is defined as the sum of the inner products in all $C_0$ components.
\end{lemma}
Lemma \ref{space_decomposition} shows that the space of games can be decomposed into three subspaces: $\mathscr{P}$, $\mathscr{H}$, and $\mathscr{N}$, and these three subspaces are orthogonal under the inner product defined in \eqref{game_inner}. Note that the direct sum of $\mathscr{P}$ and $\mathscr{N}$, i.e., $\mathscr{P}\oplus\mathscr{N}$, coincides with the set of all potential games, i.e., the games that satisfy \eqref{Potential_def}, and the direct sum of $\mathscr{H}$ and $\mathscr{N}$, i.e., $\mathscr{H}\oplus\mathscr{N}$, coincides with the set of all harmonic games, i.e., the games that satisfy \eqref{harmonic}. That is:
\begin{equation} \nonumber     \mathscr{G}_{\mathcal{N},\mathcal{A}}=\mathrlap{\underbrace{\phantom{\mathscr{P}\quad\oplus\quad\mathscr{N}}}_{\text{Potential games}}}
      \mathscr{P}\quad\oplus\quad 
      \overbrace{\mathscr{N} \quad\oplus\quad\mathscr{H}}^{\text{Harmonic games}}
\end{equation}

With Lemma \ref{space_decomposition}, given an arbitrary game $\mathcal{G}=\{\mathcal{N},\mathcal{A},\{J_i\}_{i\in\mathcal{N}}\}\in\mathscr{G}_{\mathcal{N},\mathcal{A}}$, one can project it onto the subspace $\mathscr{P}\oplus\mathscr{N}\in\mathscr{G}_{\mathcal{N},\mathcal{A}}$ to find its ``closest" potential game with respect to the inner product defined in \eqref{game_inner} to approximate it. We denote the projected game as $\hat{\mathcal{G}}=\{\mathcal{N},\mathcal{A},\{\hat{J}_i\}_{i\in\mathcal{N}}\}\in\mathscr{G}_{\mathcal{N},\mathcal{A}}$. Such a projection process can be realized  by solving a least-square optimization problem:
\begin{equation}\label{least_square}
\begin{split}
    &\min_{F,{\hat{J}_i}}\sum_{i\in{\mathcal{N}}}h_i\|J_i-\hat{J}_i\|^2\\
    &s.t. \quad D_i\hat{J}_i=D_iF, \quad \forall i\in\mathcal{N}.
    \end{split}
\end{equation}
where $F$ is the potential function of the potentialized game,  $D_i$ is the difference operator defined in \eqref{operator_difference}, and the constraint in \eqref{least_square} is to enforce the potential game condition as specified in \eqref{potential_def_alter}. The detailed algorithm for projecting an arbitrary game onto the potential game space is described in Algorithm \ref{A3}.
\begin{algorithm}[t]
\caption{Potentialize a game} \label{A3}
\hspace*{0.0in} {\bf Input:} \\ 
\hspace*{0.2in}The game to potentialize   $\mathcal{G}=\{\mathcal{N},\mathcal{A},\{J_i\}_{i\in\mathcal{N}}\}$;  \\
\hspace*{0.02in} {\bf Output:} \\
\hspace*{0.2in}Potentialized game: $\hat{\mathcal{G}}=\{\mathcal{N},\mathcal{A},\{\hat{J}_i\}_{i\in\mathcal{N}}\}$ with the potential function $F$.\\
\hspace*{0.02in} {\bf Procedures:} 
\begin{algorithmic}[1]
\STATE  Find $\hat{J}_i(\mathbf{a})$ and $F(\mathbf{a})$  $\forall\mathbf{a}\in\mathcal{A}$ and $\forall i\in\mathcal{N}$ by solving the following least-squares problem:
\begin{equation}\label{least_square_al}
\begin{split}
    &\qquad\min_{F(\mathbf{a}),{\hat{J}_i}(\mathbf{a})}\sum_{i\in{\mathcal{N}}}h_i\|J_i(\mathbf{a})-\hat{J}_i(\mathbf{a})\|^2\\
    &s.t. \quad (D_iF)(\mathbf{a},\mathbf{b})=(D_i\hat{J}_i)(\mathbf{a},\mathbf{b}), \quad \forall i\in\mathcal{N}, \forall \mathbf{a},\mathbf{b}\in \mathcal{A}.
    \end{split}
\end{equation}
\end{algorithmic}
\end{algorithm}

The above game projection approach provides a possible solution to approximating any  game, which may or may not be a potential game, using its closest potential game. If the given game is a potential game, then the projected game is the same as the original. If the original game is not a potential game, then the resulting game differs from the original in terms of agents' costs, i.e., $\exists i\in\mathcal{N}$, such that $J_i\neq \hat{J}_i$. Such a cost function deviation may lead to the solution deviation, i.e., the PSNE in the potentialized game may not be a PSNE in the original game. The robustness analysis in Section \ref{Sec_problemformulation} assesses the impact of cost function deviation on the solution deviation.

\section{Robustness Analysis}\label{sec_robustness}
This section considers the robustness margin of a game if perturbations exist in agents' cost functions. These perturbations may be caused by the game projection as described in Section \ref{sec_projection}, or by the ego vehicle's incomplete information of other agents' cost functions as specified in Challenge 5 in Section \ref{Sec_problemformulation}. We do not distinguish the sources of the perturbations here, but are only interested in the effect of the perturbations on the optimality of the solution. This part of the study is inspired by \cite{robustness}, but we use  different robustness margin definitions and analysis here to meet our specific application needs. 

Consider two games $\mathcal{G}=\{\mathcal{N},\mathcal{A},\{J_i\}_{i\in\mathcal{N}}\}\in\mathscr{G}_{\mathcal{N},\mathcal{A}}$ and $\hat{\mathcal{G}}=\{\mathcal{N},\mathcal{A},\{\hat{J}_i\}_{i\in\mathcal{N}}\}\in\mathscr{G}_{\mathcal{N},\mathcal{A}}$. Let $\hat{\mathcal{G}}$ represent the game solved by the ego vehicle. Then the game $\mathcal{G}$ can represent the original game before the game projection and/or the game associated with other agents' actual cost functions.  Define $\Delta J_i:\mathcal{A}\rightarrow\mathbb{R}$ as
\begin{equation}\label{delta_J}
    \Delta J_i=J_i-\hat{J}_i.
\end{equation}
\begin{defn} [Robustness Margin] \label{def_robustnessmargin} 
Consider the games $\mathcal{G}$ and $\hat{\mathcal{G}}$. Let $\Delta J_i$ be given by \eqref{delta_J}. The robustness margin of a PSNE $\hat{\mathbf{a}}^*$ solved from the game $\hat{\mathcal{G}}$ (denoted as $\mu(\hat{\mathbf{a}}^*)$) is defined as the supremum of the magnitude $\|\Delta J\|_{\infty}$ such that $\hat{\mathbf{a}}^*$ is also a PSNE of the game $\mathcal{G}$, where 
\begin{equation}\label{delta_J_norm}
    \|\Delta J\|_{\infty}=\max_{i\in\mathcal{N}}\|\Delta J_i\|_{\infty},
\end{equation}
\begin{equation}\label{Delta_j_i_norm}
    \|\Delta J_i\|_{\infty}=\max_{\mathbf{a}\in\mathcal{A}}|\Delta J_i(\mathbf{a})|.
\end{equation} 
\end{defn}

\begin{thm}\label{lemma_robustness margin}
Consider the games $\mathcal{G}$ and $\hat{\mathcal{G}}$ and the robustness margin defined in Definition \ref{def_robustnessmargin}. Then the robustness margin $\mu(\hat{\mathbf{a}}^*)$ is 
\begin{equation}\label{robustness_margin}
    \mu(\hat{\mathbf{a}}^*)= \frac{1}{2}\min_{i\in\mathcal{N}}\chi_i(\hat{\mathbf{a}}^*),
\end{equation}
where $\chi_i:\mathcal{A}\rightarrow \mathbb{R}$ is given by
\begin{equation}\label{chi}
\chi_i(\mathbf{a})=\min_{(\mathbf{a},\mathbf{b})\in\mathcal{C}_i, \mathbf{a}\neq \mathbf{b}}\left(\hat{J}_i(\mathbf{b})-\hat{J}_i(\mathbf{a})\right).
\end{equation}
\end{thm}
\begin{proof}
    To show that \eqref{robustness_margin} holds, in the first step, we show that 
\begin{equation}\label{robustness_margin_min}
    \mu(\hat{\mathbf{a}}^*)\geq \frac{1}{2}\min_{i\in\mathcal{N}}\chi_i(\hat{\mathbf{a}}^*).
\end{equation}
In the second step, we show that
\begin{equation}\label{robustness_margin_max}
    \mu(\hat{\mathbf{a}}^*)\leq \frac{1}{2}\min_{i\in\mathcal{N}}\chi_i(\hat{\mathbf{a}}^*).
\end{equation}

\textbf{\textit{Step 1: To prove \eqref{robustness_margin_min}}.}

For every $\{\Delta J_i\}_{i\in\mathcal{N}}$ that satisfies $\|\Delta J\|_{\infty}\leq \frac{1}{2}\min_{i\in\mathcal{N}}\chi_i(\hat{\mathbf{a}}^*)$, it must satisfy
\begin{equation}\label{35_1030}
    \|\Delta J_i\|_{\infty}\leq \frac{1}{2}\chi_i(\hat{\mathbf{a}}^*), \quad \forall i\in\mathcal{N}.
\end{equation}
Eq. \eqref{35_1030} leads to
    \begin{equation}
    \begin{split}
      \Delta J_i(\hat{\mathbf{a}}^*)-  &\Delta J_i(\mathbf{b})\leq 2\|\Delta J_i\|_{\infty}\leq \hat{J}_i(\mathbf{b})-\hat{J}_i(\hat{\mathbf{a}}^*),\\
      & \forall\mathbf{b}\neq \hat{\mathbf{a}}^* \text{ and } (\hat{\mathbf{a}}^*,\mathbf{b})\in\mathcal{C}_i. 
    \end{split}
    \end{equation}
Therefore, the following inequality holds:  
\begin{equation}\label{32}
        \Delta J_i(\hat{\mathbf{a}}^*)+\hat{J}_i(\hat{\mathbf{a}}^*)\leq \Delta J_i(\mathbf{b})+\hat{J}_i(\mathbf{b}).
\end{equation}
Combining \eqref{delta_J} and \eqref{32}, we have
\begin{equation}\label{33}
    J_i(\hat{\mathbf{a}}^*) \leq J_i(\mathbf{b}).
\end{equation}
Because \eqref{33} holds $\forall i\in\mathcal{N}$ and $\forall \mathbf{b}\neq \hat{\mathbf{a}}^*$, $(\hat{\mathbf{a}}^*,\mathbf{b})\in\mathcal{C}_i$, according to Definition \ref{d2}, $\hat{\mathbf{a}}^*$ is a PSNE for the game $\mathcal{G}$, completing the proof of \eqref{robustness_margin_min}.

\textbf{\textit{Step 2: To prove \eqref{robustness_margin_max}}.}

Let us construct a $\{\Delta J_i\}_{i\in\mathcal{N}}$  with $\|\Delta J\|_{\infty}=\frac{1}{2}\min_{i\in\mathcal{N}}\chi_i(\hat{\mathbf{a}}^*)+\epsilon$, where $\epsilon>0$ is an arbitrary small number. We show that with the constructed $\{\Delta J_i\}_{i\in\mathcal{N}}$,  the strategy profile $\hat{\mathbf{a}}^*$ is not a PSNE of the game $\mathcal{G}$, and thus, the inequality \eqref{robustness_margin_max} holds.

Consider the game $\hat{\mathcal{G}}$. Let $j\in\mathcal{N}$ be the agent that achieves the maximum on the right-hand side of \eqref{delta_J_norm}. Let $\mathbf{b}\in\mathcal{A}$ be the startegy profile that achieves the minimum on the right-hand side of \eqref{chi} with $\mathbf{a}$ being $\hat{\mathbf{a}}^*$. Construct a $\{\Delta J_i\}_{i\in\mathcal{N}}$ such that 
\begin{equation}\label{38}
    \Delta J_j(\hat{\mathbf{a}}^*)= \frac{1}{2}\left(\hat{J}_j(\mathbf{b})-\hat{J}_j(\hat{\mathbf{a}}^*)\right)+\epsilon,
\end{equation}
\begin{equation}\label{39}
    \Delta J_j(\mathbf{b})=-\Delta J_j(\hat{\mathbf{a}}^*).
\end{equation}
\begin{equation}
    \Delta J_i(\mathbf{a})=0, \forall i\neq j, \text{ and } \forall \mathbf{a}\notin \{\hat{\mathbf{a}}^*,\mathbf{b}\}.
\end{equation}
Equations \eqref{38} and \eqref{39} lead to
\begin{equation}
\begin{split}
     \Delta J_j(\hat{\mathbf{a}}^*)-\Delta J_j(\mathbf{b})&=2\Delta J_j(\hat{\mathbf{a}}^*)\\
     &=\left(\hat{J}_j(\mathbf{b})-\hat{J}_j(\hat{\mathbf{a}}^*)\right)+2\epsilon.
     \end{split}
\end{equation}
Therefore, we have
\begin{equation}\label{41}
   \Delta J_j(\hat{\mathbf{a}}^*)+\hat{J}_j(\hat{\mathbf{a}}^*)=\Delta J_j(\mathbf{b})+\hat{J}_j(\mathbf{b})+2\epsilon. 
\end{equation}
Combining \eqref{delta_J} and \eqref{41}, we have 
\begin{equation}\label{42}
J_j(\hat{\mathbf{a}}^*)=J_j(\mathbf{b})+2\epsilon.
\end{equation}
Thus, 
\begin{equation}\label{42}
J_j(\mathbf{b})<J_j(\hat{\mathbf{a}}^*).
\end{equation}
Note that $(\hat{\mathbf{a}}^*,\mathbf{b})\in\mathcal{C}_j$. Therefore, $\hat{\mathbf{a}}^*$ is not a PSNE of the game $\mathcal{G}$, completing the proof of \eqref{robustness_margin_max}.

The inequalities \eqref{robustness_margin_min} and \eqref{robustness_margin_max} together  lead to \eqref{robustness_margin}, completing the proof of Theorem \ref{lemma_robustness margin}. 
\end{proof}

With the robustness margin in Theorem \ref{lemma_robustness margin}, after projecting a game onto the potential game space using Algorithm \ref{A3} and deriving the PSNE of the projected game using Algorithm \ref{A2}, one may check whether $\|\Delta J\|_{\infty}\leq \frac{1}{2}\min_{i\in\mathcal{N}}\chi_i(\hat{\mathbf{a}}^*)$ holds. If yes, then this NE is also a NE of the original game, despite that the original game may not be a potential game. We shall provide a specific numerical example in the next section.

On the aspect of incomplete information, Theorem \ref{lemma_robustness margin} quantifies how "trustworthy" the game-theoretic solution is in autonomous driving, when facing various interactive agents. Specifically, although it is inevitable that the assumed other agents' cost functions deviate from their actual ones, such a deviation may not necessarily lead to a solution deviation. A game with a large robustness margin suggests that it can accommodate a wide range of agents' costs, e.g., agents of a wide range of aggressiveness, thereby indicating a trustworthy solution against the variability of other agents' driving styles/preferences.

Theorem \ref{lemma_robustness margin} also reveals that the robustness margin of a PSNE $\hat{\mathbf{a}}^*$ of a game $\hat{\mathcal{G}}$ depends on $\min_{(\hat{\mathbf{a}}^*,\mathbf{b})\in\mathcal{C}_i, \hat{\mathbf{a}}^*\neq \mathbf{b}}\left(\hat{J}_i(\mathbf{b})-\hat{J}_i(\hat{\mathbf{a}}^*)\right)$. Intuitively, it indicates that given a specific cost function design (e.g., \eqref{cost_1}) and the range of agent strategy  (e.g., the acceleration variation range), the sparser the strategy space, the more robust the PSNE of the game is. That is, for example, if a vehicle's strategy is its acceleration, then the game with the strategy space $[-3, -1, 1, 3]$  $m/s^2$ would result in a larger robustness margin than that with $[-3, -2, -1, 0, 1, 2, 3]$  $ m/s^2$ as strategy space. This observation indicates that game-theoretic approaches might be more reliable (in the sense of a larger robustness margin) in predicting other agents' high-level intents (e.g., whether a vehicle would yield or proceed) compared to predicting their accurate maneuvers (e.g., whether a vehicle brakes at $-2$  $m/s^2$ or $-3$  $m/s^2$). This observation is consistent with human driving experience. When a human driver interacts with others, e.g., to merge onto a highway with dense traffic, based on the common driving experience, they often  ``predict" whether a vehicle is going to allow them to merge/proceed in front of it, as opposed to trying to accurately estimate such a vehicle acceleration. 

In the case where the maximum magnitude of the cost function deviation exceeds the robustness margin in \eqref{robustness_margin}, the PSNE of $\mathcal{G}$ may not necessarily be a PSNE of $\hat{\mathcal{G}}$ any more, but it can still be an $\epsilon$-NE in $\hat{\mathcal{G}}$ (see Definition \ref{epsilonNash_de} for $\epsilon$-NE). The value of $\epsilon$ depends on the distance of the two games and is given by the following lemma. 

\begin{lemma}[Theorem 6.3 in \cite{game_decomposition}]\label{accuracy}
    Let $\mathcal{G}\in\mathscr{G}_{\mathcal{N},\mathcal{A}}$ and $\hat{\mathcal{G}}\in\mathscr{G}_{\mathcal{N},\mathcal{A}}$ be two games. Define $\alpha\triangleq \|\mathcal{G}-\hat{\mathcal{G}}\|_{\mathcal{N},\mathcal{A}}$. Then every PSNE of $\hat{\mathcal{G}}$ is an $\epsilon$-NE of $\mathcal{G}$ for some $\epsilon\leq \max_{i\in\mathcal{N}}\frac{2\alpha}{\sqrt{h_i}}$.
\end{lemma}

Lemma \ref{accuracy} indicates that although the cost function deviation may lead to the solution deviation, we can  quantify the optimality of the obtained solution in the original game. 

\section{Numerical Studies}\label{sec:simulation}
This section reports the outcome from numerical studies to illustrate and to verify the algorithms and theoretical analysis in Sections \ref{Sec_problemformulation}-\ref{sec_robustness}.

\subsection{Scenario setup}
We consider a highway lane-changing scenario as pictured in Figure \ref{scenario}. The vehicles' dynamics are described by a kinematic bicycle model \cite{bicycle_2,bicycle_validation}:
\begin{equation}\label{dynamics}
\begin{split}
    x_i(t+1)&=x_i(t)+   v_{i}(t)\text{cos}(\phi_i(t)+\beta_i(t))\Delta t,\\
    y_i(t+1)&=y_i(t)+ v_{i}(t)\text{sin}(\phi_i(t)+\beta_i(t))\Delta t,\\
    v_i(t+1)&=v_i(t)+ w_{i}(t)\Delta t,\\
    \phi_{i}(t+1)&=\phi_{i}(t)+ \frac{v_i(t)}{l_r}\text{sin}(\beta_i(t))\Delta t,\\
    \beta_i(t+1)&=\beta_i(t)+\text{tan}^{-1}\left(\frac{l_r}{l_r+l_f}\text{tan}(\delta_{i,f}(t))\right)\Delta t,
    \end{split}
\end{equation}
where $i=1,2,3$ designates the $i^{th}$ vehicle; $\Delta t=0.5s$ is the sampling time; $x_i$ and $y_i$ are the longitudinal and lateral position of the center of mass of  vehicle $i$ along $x$ and $y$ axes, respectively; $v_{i}$ and  $w_{i}$  are  the velocity and acceleration of the center of mass of  vehicle $i$, respectively; $\beta_i$ is  the angle of the velocity with respect to the longitudinal axis of the vehicle; $\phi_i$ is the inertial heading; $l_f$ and $l_r$ are the lengths from the center of mass to the front and rear ends of the car, respectively, and are selected to be $l_f = l_r = 1.5$ $m$; $\delta_{i,f}$ is the steering angle of the front wheels, and the rear wheel steering angles are assumed to be zero since for many vehicles the
rear wheels cannot be steered \cite{bicycle_2}. The inputs of each vehicle are the acceleration $w_i(t)$ and the steering angle $\delta_{i,f}(t)$.

\begin{figure}[thpb]
\centering
\includegraphics[width=0.4\textwidth]{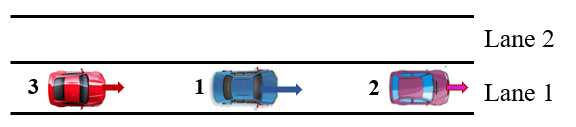}
\caption{Highway lane-changing scenario.}\label{scenario}
\end{figure}

In the pictured scenario in Figure \ref{scenario}, vehicles aim to select an appropriate lane to drive in to keep their desired speeds. In this case, the action space of each vehicle is set to be  \begin{equation}\label{finite_space}  \mathcal{U}_i=\{\text{Lane 1, Lane 2}\}.
\end{equation}
After a lane is selected, a vehicle needs to determine whether it needs to perform a lane-changing, i.e., whether the selected lane (called target lane) differs from its current lane. If the lane-changing is needed, then a constant steering angle  $\delta_{i,f}(t)=\pm0.9^{\circ}$ is applied until the vehicle is in the center of the target lane, after which a lane centering controller is activated to reduce the relative heading angle between the vehicle and the road to zero. The longitudinal acceleration $w_i(t)$ of each vehicle remains zero during the process. Initially, all vehicles are placed in Lane 1, with their relative positions shown in Figure \ref{scenario}. 

Since the lane-changing maneuver often takes around $4s$ to complete \cite{add_T}, the prediction horizon length $T$ is selected to correspond to $4s$. The strategy space $\mathcal{A}_i$, $i=1,2,3$, is constructed based on the action space $\mathcal{U}_i$ such that a constant action, e.g., Lane 1, is planned over the prediction horizon, i.e.,  $a_i(\tau)=a_i(t)\in\mathcal{U}_i$, $\forall \tau\in[t,t+T-1]$ in \eqref{optimization}. Note that although a constant action is planned, a vehicle may change its mind after $\Delta t$, i.e., a vehicle may terminate an attempt of lane-changing and change back to its previous lane if needed. This is triggered by the receding-horizon game setting in \eqref{optimization}.  
\subsection{Cost function setup}
The  cumulative cost of each vehicle is {designed} as
\begin{equation}\label{cost_1}
\begin{split}
    &J_i(\mathbf{a}_i(t),\mathbf{a}_{-i}(t))\\
    &=J_i^{self}(\mathbf{a}_i(t))+\sum_{j\in \mathcal{N},j\neq i}J_{ij}(\mathbf{a}_i(t),\mathbf{a}_j(t))\\
    &\quad +\sum_{j\in\mathcal{N}_{-i}^*}\mathbf{1}_{\text{rules}}(\mathbf{a}_i(t),\mathbf{a}_j(t))
\end{split}
\end{equation}
where the first term $J_i^{self}(\mathbf{a}_i(t))$ is to discourage lane changes, and a result, to make lane-changing less frequent. Specifically,  $J_i^{self}(\mathbf{a}_i(t))$ is designed as 
\begin{align}
&J_i^{self}(\mathbf{a}_i(t))=\begin{cases} 1, &\mbox{if $\mathbf{a}_i(t)\neq \mathbf{a}_i^*(t-1)$},\vspace{1ex}\\
0, &\mbox{otherwise}.\vspace{1ex}\end{cases}
\end{align}
Note that $\mathbf{a}_i^*(t-1)$ represents vehicle $i$'s selected lane at $t-1$, and $\mathbf{a}_i(t)\neq \mathbf{a}_i^*(t-1)$ indicates that a lane-changing is required to perform $\mathbf{a}_i(t)$.

The second term $\sum_{j\in \mathcal{N},j\neq i}J_{ij}(\mathbf{a}_i(t),\mathbf{a}_j(t))$ is to avoid  collision:
\begin{equation}\label{collision_design}
\begin{split}
    &\sum_{j\in \mathcal{N},j\neq i}J_{ij}(\mathbf{a}_i(t),\mathbf{a}_j(t))\\
    &=\sum_{\tau=t}^{t+T-1}C_{ij}(x_i(\tau),y_i(\tau),x_j(\tau),y_j(\tau)),
    \end{split}
\end{equation}
where
{\begin{equation}\label{tanh_cost}
\begin{split}
    &C_{ij}(x_i(\tau),y_i(\tau),x_j(\tau),y_j(\tau))\\
    &=\left(\tanh{\left(\beta(d^2_{x,c}-\left(x_i(\tau)-x_j(\tau)\right)^2)\right)}+1\right)\\
    &\quad\cdot\left(\tanh{(\beta(d^2_{y,c}-\left(y_i(\tau)-y_j(\tau)\right)^2))}+1\right).
    \end{split}
\end{equation}
Here $d_{x,c}$ and $d_{y,c}$ are the longitudinal and lateral collision distances, respectively, and we select $d_{x,c}=7$ $m$ and $d_{y,c}=4.5$ $m$ in the simulation. The parameter $\beta$ is selected to be sufficiently large such that the $tanh$ function always takes the two extreme values $-1$ or $1$. The cost \eqref{tanh_cost} means that at time $\tau$, if vehicle $j$ is a collision threat to vehicle $i$ both laterally and longitudinally, then $J_{ij}(x_i(\tau),y_i(\tau),x_j(\tau),y_j(\tau))=4$; Otherwise, $J_{ij}(x_i(\tau),y_i(\tau),x_j(\tau),y_j(\tau))=0$. Note that the lane width of our simulation environment is set to be $5$ $m$, and {the lateral collision distance $d_{y,c}$ is selected to be a bit smaller than the road width such that if a vehicle is not in the same lane or in the target lane  of the ego vehicle, then it is not considered as a threat to the ego vehicle.

The third term $\sum_{j\in\mathcal{N}_{-i}^*}\mathbf{1}_{\text{rules}}(\mathbf{a}_i(t),\mathbf{a}_j(t))$ is to encourage the compliance with right-of-way/road priority rules, which, in our simulation, are that the rear vehicle (resp. the lane-changing vehicle) has the responsibility to take actions to avoid rear-end accidents (resp. accidents during the lane-changing). 
Here $\mathcal{N}_{-i}^*$ is the set of agents with higher road priority than agent $i$. It includes the vehicle right in front of $i$ and the ones in vehicle $i$'s target lane if $\mathbf{a}_i(t)$ requires lane-changing. The notation  $\mathbf{1}_{\text{condition}}$ is defined as
\begin{align}\label{1}
&\mathbf{1}_{\text{condition}}=\begin{cases} 1, &\mbox{if condition is true},\vspace{1ex}\\
0, &\mbox{otherwise}.\vspace{1ex}\end{cases}
\end{align}

Note that with the last term in the cost function (i.e., $\sum_{j\in\mathcal{N}_{-i}^*}\mathbf{1}_{\text{rules}}(\mathbf{a}_i(t),\mathbf{a}_j(t))$), the agents' interactions are not symmetric anymore, i.e., $J_i$ does not satisfy the condition required in Theorem \ref{t1}. We shall show in the next subsection that this formulated game is not a potential game.
\subsection{Game solution, projection, and robustness}
Let us consider the formulated game: 
\begin{equation}\label{game_simulation}
    \begin{split}
    \mathbf{a}^*_i(t)&\in \argmin_{\mathbf{a}_i(t)\in\mathcal{A}_i}J_i(\mathbf{a}_i(t),\mathbf{a}_{-i}(t)),
    \end{split}
\end{equation}
where $i\in\{1,2,3\}$ and $J_i$ is given by \eqref{cost_1}. Consider the scenario where vehicles $2$ and $3$ drive at the same speed, faster than vehicle $1$. Initially (i.e., $t=0$ $s$), all vehicles are in Lane 1 as shown in Figure \ref{scenario} with safe inter-vehicle distances, i.e., $J_{ij}(\mathbf{a}_i(0),\mathbf{a}_j(0))=0,\forall i,j\in\{1,2,3\}, i\neq j$, and $\forall \mathbf{a}_i(0), \mathbf{a}_j(0) \in\mathcal{A}$. Because vehicle $3$ drives faster than vehicle $1$, $\Delta 
 x_{1,3}(t)=|x_1(t)-x_3(t)|$ decreases with time. Consider the smallest $t_{c}>0$ such that, if no lane-changing, then $J_{13}(\mathbf{a}_1(t_{c}),\mathbf{a}_3(t_{c}))\neq 0$. This $t_{c}$ corresponds to the earliest time such that if no action is taken, then a collision would happen at the end of the prediction horizon, $t_{c}+T-1$, i.e., $\Delta 
 x_{1,3}(t_{c}+T-1)<d_{x,c}$. In our simulation, $t_c$ approximately corresponds to $4s$, as shown in Figures \ref{PSNE_4} and \ref{PSNE_6}. Let us analyze the game \eqref{game_simulation} at this key time point. We denote this game as $\mathcal{G}$.

 Since every vehicle has two possible strategies: Lane 1 or Lane 2, the total number of strategy profiles is $2^3=8$. For ease of representation, we label these $8$ strategy profiles using the numbers between $1$ and $8$ as shown in Table \ref{Table I}. In a strategy profile, the $k^{th}$ element represents the $k^{th}$ vehicle's strategy. For example, the strategy profile $(2,2,1)$ (which is labeled as ``$7$" in Table \ref{Table I}) represents that vehicles $1$ and $2$ select Lane 2, and vehicle $3$ selects Lane 1.

\begin{table}[!h]
\centering
\caption{Strategy profile labels}\label{Table I}
\begin{tabular}{c|c|c|c}
\hline  
 Strategy profile & Label & Strategy profile & Label
\\
\hline
    $(1,1,1)$ & $\mathbf{1}$ &  $(2,1,2)$ & $\mathbf{5}$\\
\hline
    $(1,2,1)$ & $\mathbf{2}$ &  $(2,1,1)$ & $\mathbf{6}$\\
\hline
    $(1,2,2)$ & $\mathbf{3}$ &  $(2,2,1)$ & $\mathbf{7}$\\
\hline
    $(1,1,2)$ & $\mathbf{4}$ &  $(2,2,2)$ & $\mathbf{8}$\\
\hline
\end{tabular}
\end{table}
\begin{figure*}[thpb]
\centering
\subfigure[]{\label{original}
\includegraphics[width=0.27\textwidth]{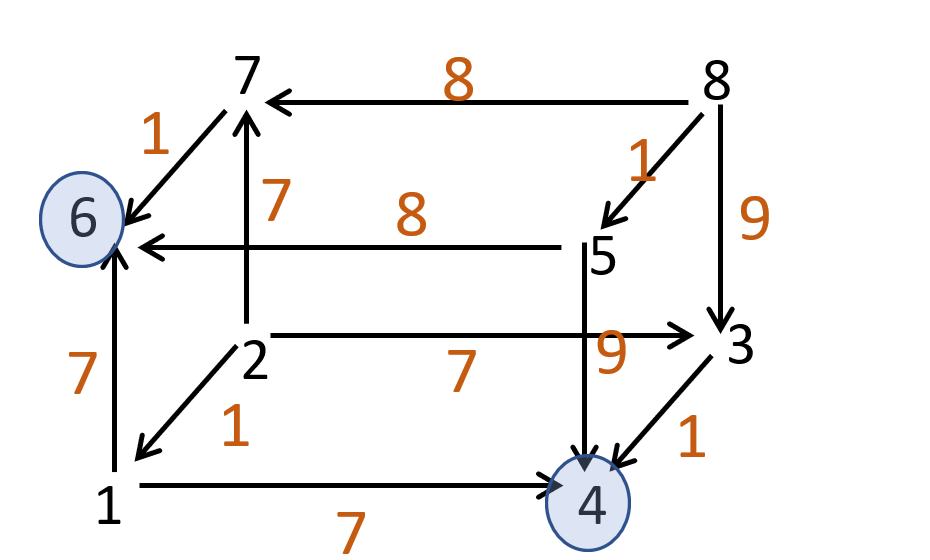}}
\subfigure[]{\label{projected}
\includegraphics[width=0.27\textwidth]{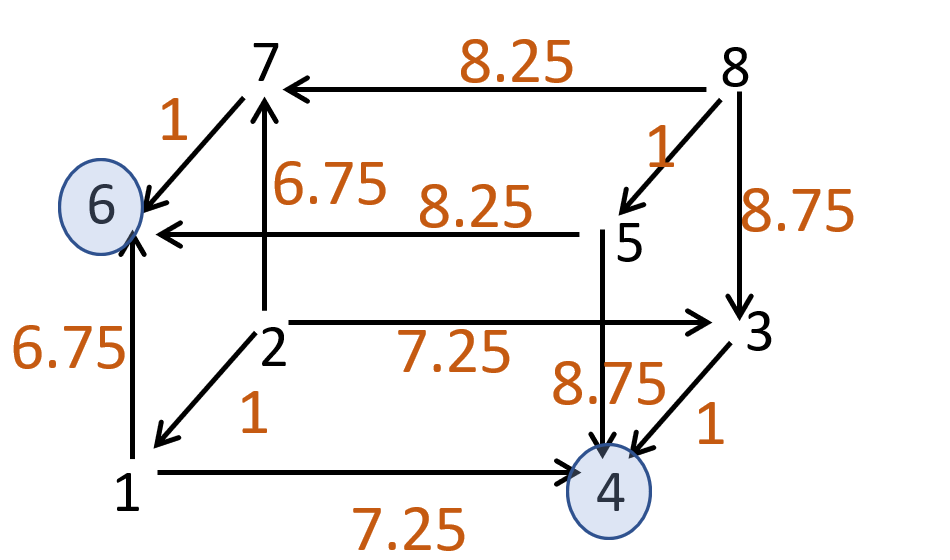}}
\subfigure[]{\label{potentialfunction}
\includegraphics[width=0.27\textwidth]{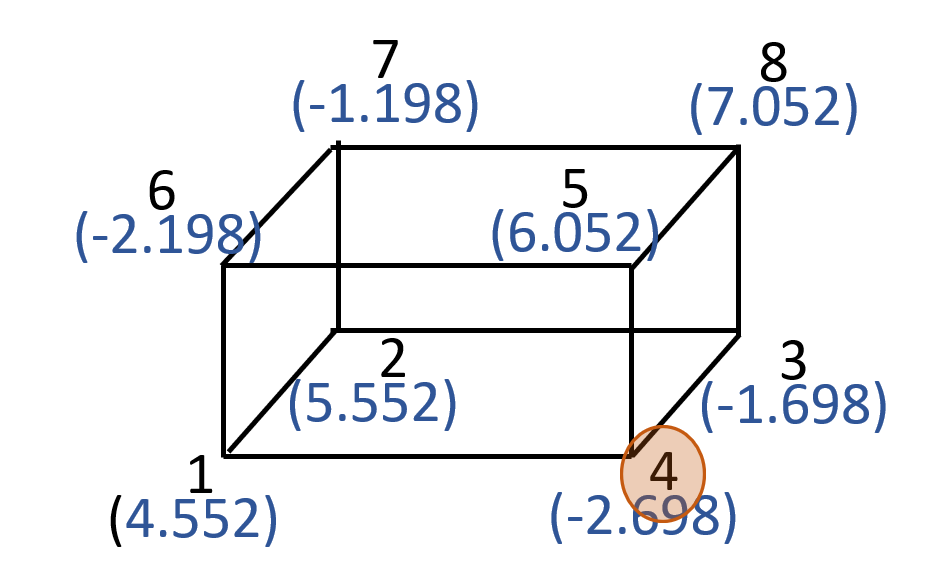}}
{\caption{(a) The flow representation of the game $\mathcal{G}$; (b) The flow representation of the potentialized game $\hat{\mathcal{G}}$; and (c) The potential function value of $\hat{\mathcal{G}}$ at each strategy profile.}}\label{lanechanging_1}
\end{figure*}
\begin{figure*}
    \centering
    \includegraphics[width=0.85\textwidth]{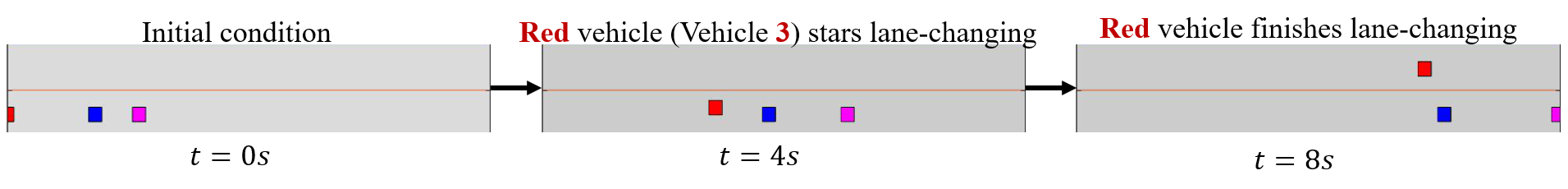}
    \caption{Vehicle lane-changing behaviors at the strategy profile 4. }
    \label{PSNE_4}
\end{figure*}
\begin{figure*}
    \centering
    \includegraphics[width=0.85\textwidth]{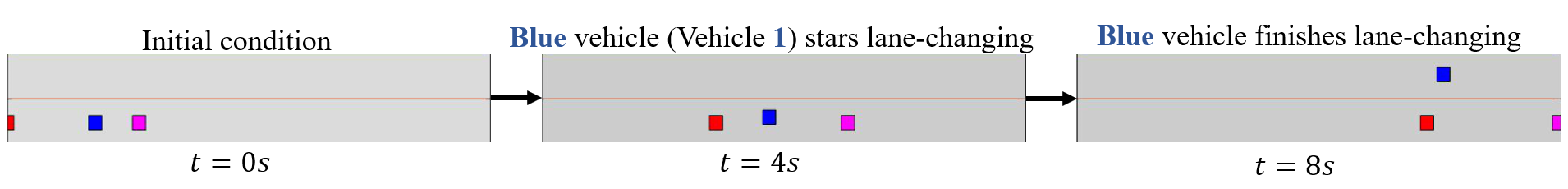}
    \caption{Vehicle lane-changing behaviors at the strategy profile 6.}
    \label{PSNE_6}
\end{figure*}
With these labels, we represent the game $\mathcal{G}$ using its flow representation as shown in Figure \ref{original}. Here the nodes between $1$ and $8$ represent the $8$ strategy profiles. Each edge connects two comparable strategy profiles, and the edge flow is defined by the pairwise comparison function \eqref{X}. For example, the flow from the strategy profile ``1" to ``4" means that vehicle $3$'s cost will be reduced by 1 if changing from the strategy profile ``1" to ``4" (note that strategy profiles ``1" and ``4" differ only in vehicle $3$'s strategy). With such a flow representation, the PSNE of the game is clear: The node that has no flow ``leaving" this node is a PSNE. It is because according to Definition \ref{Nash_de}, a PSNE is a strategy profile that all its comparable strategy profiles are equally or less preferred than it, i.e., no cost reduction is possible by unilateral strategy deviation. With this criterion, we can observe that the game shown in  Figure \ref{original} has two PSNE: The strategy profiles 4 and 6. They are marked with the blue circles.

The strategy profile 4 corresponds to (1,1,2), representing that vehicle $3$ moves to Lane 2 while the other two stay in Lane 1. This NE is found in our simulation when we use the best response algorithm (i.e., Algorithm \ref{A_br}) to solve the game. The vehicles' behavior at three important time points is illustrated in Figure \ref{PSNE_4}. The strategy profile 6 corresponds to  (2,1,1), representing that vehicle $1$ moves to Lane 2 while the other two stay in Lane 1. This equilibrium is also found in our simulation by the best response algorithm, and its corresponding vehicles' behavior can be visualized in Figure \ref{PSNE_6}. Since both strategy profiles 4 and 6 are PSNE, they are both individually optimal solutions and are both able to resolve the potential collision conflict between vehicles 1 and 3 with the minimum lane-changing effort. However, we also note that these two PSNE are not equally preferred: Considering the road priority rule, vehicle $3$ should perform the lane-changing to avoid the rear-end accident, resulting in that strategy profile 4 is preferable to 6 considering the global situation.

Next let us project the game $\mathcal{G}$ in Figure \ref{original} onto the potential game space using Algorithm \ref{A3}. The potentialized game (denoted as $\hat{\mathcal{G}}$) has the flow representation as shown in Figure \ref{projected}. It is clear that the flows in Figure \ref{projected} do not coincide with that in Figure \ref{original}, indicating that the original game is not a potential game. 

For the projected game $\hat{\mathcal{G}}$, it also has two PSNE: Strategy profiles 4 and 6. For each PSNE, according to Theorem \ref{lemma_robustness margin}, its robustness margin can be calculated as 
\begin{equation}
\begin{split}
    \mu(\hat{\mathbf{a}}^*)= &\frac{1}{2}\min_{i\in\mathcal{N}}\min_{(\mathbf{a},\mathbf{b})\in\mathcal{C}_i, \mathbf{a}\neq \mathbf{b}}\left(\hat{J}_i(\mathbf{b})-\hat{J}_i(\mathbf{a})\right)\\
    =&0.5
    \end{split}
    \end{equation}
By calculating $\Delta J_i=J_i-\hat{J}_i$, where  $\hat{J}_i$ is the cost of the potentialized game derived by Algorithm \ref{A3},  we can find that the game projection leads to $\|\Delta J\|_{\infty}=0.125$. Because $\|\Delta J\|_{\infty}<\mu(\hat{\mathbf{a}}^*)$, we know that such a game projection would not lead to solution deviation, i.e., all the PSNE of $\hat{\mathcal{G}}$ remain a PSNE of $\mathcal{G}$. This result is confirmed by observing the PSNE of the games from their flow representations in Figures  \ref{original} and \ref{projected}, respectively.  Therefore, although the original game is not a potential game and that the potentialized game differs from the original game (i.e., cost deviations exist), the sets of solutions of the two games coincide.

The potential function $F$ of the potentialized game $\hat{\mathcal{G}}$ associated with each strategy profile is shown in Figure \ref{potentialfunction}. We can observe from this figure that the PSNE that minimizes the potential function is unique and is the strategy profile 4 (marked with the orange circle). Therefore, strategy profile 4 is not only individually optimal (in the sense of NE) but also socially optimal (in the sense of optimizing the ``team interest" characterized by the potential function), consistent with our analysis above that strategy profile 4 is preferable to 6. Moreover, to solve the potentialized game $\hat{\mathcal{G}}$, one can use the potential function optimization algorithm (i.e., Algorithm \ref{A2}), which requires one optimization as opposed to iterative optimizations in the best response algorithm (i.e., Algorithm \ref{A_br}). The above analysis confirms the merits of the potential/potentialized game. 

We also report the running time of the game projection. The running time is collected from MATLAB$^{\circledR}$ on a laptop with an Intel Core i9-12900H processor clocked at $2.50$ GHz and $32$ GB of RAM. For each game projection, i.e., each execution of Algorithm \ref{A3} to find the potentialized game $\hat{\mathcal{G}}$ and the associated potential function $F$, the average running time is $ 0.05s$,  confirming that it's practical to use game projection for real-time operation. 

\section{Conclusion}\label{sec_conclusion}
This paper developed a more practical and robust game-theoretic framework for autonomous driving. Challenges caused by game complexity and incomplete information were addressed. A receding-horizon multi-player game was formulated to characterize agents interactions for AV decision-making. Recognizing the difficulty in solving non-potential games, we employed game decomposition techniques to effectively project arbitrary games onto potential game space. This projection, while potentially leading to cost function deviations, does not necessarily result in solution deviation. We investigated this aspect by defining and deriving robustness margins, which represent the maximum allowable cost function deviations without impacting the optimality of a game solution. Our findings, supported by numerical studies in a lane-changing traffic scenario, demonstrated that asymmetric agents interactions and non-potential games can be effectively handled using our framework within a reasonable computational time. This research broadens the scope of potential game-based approaches in autonomous driving and establishes a robustness analysis framework that quantifies solution stability in multi-player games against incomplete information in cost functions. 



\bibliography{references}

\begin{thebibliography}{10}
\providecommand{\url}[1]{#1}
\csname url@samestyle\endcsname
\providecommand{\newblock}{\relax}
\providecommand{\bibinfo}[2]{#2}
\providecommand{\BIBentrySTDinterwordspacing}{\spaceskip=0pt\relax}
\providecommand{\BIBentryALTinterwordstretchfactor}{4}
\providecommand{\BIBentryALTinterwordspacing}{\spaceskip=\fontdimen2\font plus
\BIBentryALTinterwordstretchfactor\fontdimen3\font minus \fontdimen4\font\relax}
\providecommand{\BIBforeignlanguage}[2]{{%
\expandafter\ifx\csname l@#1\endcsname\relax
\typeout{** WARNING: IEEEtran.bst: No hyphenation pattern has been}%
\typeout{** loaded for the language `#1'. Using the pattern for}%
\typeout{** the default language instead.}%
\else
\language=\csname l@#1\endcsname
\fi
#2}}
\providecommand{\BIBdecl}{\relax}
\BIBdecl

\bibitem{my_potential}
M.~Liu, I.~Kolmanovsky, H.~E. Tseng, S.~Huang, D.~Filev, and A.~Girard, ``Potential game-based decision-making for autonomous driving,'' \emph{IEEE Transactions on Intelligent Transportation Systems}, vol.~24, no.~8, pp. 8014--8027, 2023.

\bibitem{roadsafety}
D.~J. Fagnant and K.~Kockelman, ``Preparing a nation for autonomous vehicles: opportunities, barriers and policy recommendations,'' \emph{Transportation Research Part A: Policy and Practice}, vol.~77, pp. 167--181, 2015.

\bibitem{fuelefficiency}
Z.~Wadud, D.~MacKenzie, and P.~Leiby, ``Help or hindrance? the travel, energy and carbon impacts of highly automated vehicles,'' \emph{Transportation Research Part A: Policy and Practice}, vol.~86, pp. 1--18, 2016.

\bibitem{trafficcongestion}
M.~Wang, W.~Daamen, S.~P. Hoogendoorn, and B.~Van~Arem, ``Connected variable speed limits control and car-following control with vehicle-infrastructure communication to resolve stop-and-go waves,'' \emph{Journal of Intelligent Transportation Systems}, vol.~20, no.~6, pp. 559--572, 2016.

\bibitem{inhancedmobility}
J.~P. Zmud and I.~N. Sener, ``Towards an understanding of the travel behavior impact of autonomous vehicles,'' \emph{Transportation Research Procedia}, vol.~25, pp. 2500--2519, 2017.

\bibitem{productivity}
C.~D. Harper, C.~T. Hendrickson, S.~Mangones, and C.~Samaras, ``Estimating potential increases in travel with autonomous vehicles for the non-driving, elderly and people with travel-restrictive medical conditions,'' \emph{Transportation Research Part C: Emerging Technologies}, vol.~72, pp. 1--9, 2016.

\bibitem{parking}
W.~Zhang and S.~Guhathakurta, ``Parking spaces in the age of shared autonomous vehicles: How much parking will we need and where?'' \emph{Transportation Research Record}, vol. 2651, no.~1, pp. 80--91, 2017.

\bibitem{trafficflow}
A.~Talebpour and H.~S. Mahmassani, ``Influence of connected and autonomous vehicles on traffic flow stability and throughput,'' \emph{Transportation Research Part C: Emerging Technologies}, vol.~71, pp. 143--163, 2016.

\bibitem{newbusiness}
M.~K{\"o}nig and L.~Neumayr, ``Users’ resistance towards radical innovations: The case of the self-driving car,'' \emph{Transportation Research Part F: Traffic Psychology and Behaviour}, vol.~44, pp. 42--52, 2017.

\bibitem{challenge1_2022}
B.~Toghi, R.~Valiente, D.~Sadigh, R.~Pedarsani, and Y.~P. Fallah, ``Social coordination and altruism in autonomous driving,'' \emph{IEEE Transactions on Intelligent Transportation Systems}, vol.~23, no.~12, pp. 24\,791--24\,804, 2022.

\bibitem{challenge2_2022}
X.~Tang, B.~Huang, T.~Liu, and X.~Lin, ``Highway decision-making and motion planning for autonomous driving via soft actor-critic,'' \emph{IEEE Transactions on Vehicular Technology}, vol.~71, no.~5, pp. 4706--4717, 2022.

\bibitem{game_1}
Y.~Wang, Y.~Ren, S.~Elliott, and W.~Zhang, ``Enabling courteous vehicle interactions through game-based and dynamics-aware intent inference,'' \emph{IEEE Transactions on Intelligent Vehicles}, vol.~5, no.~2, pp. 217--228, 2019.

\bibitem{game_merge}
H.~Kita, ``A merging--giveway interaction model of cars in a merging section: a game theoretic analysis,'' \emph{Transportation Research Part A: Policy and Practice}, vol.~33, no. 3-4, pp. 305--312, 1999.

\bibitem{game_pursuer}
Z.~Zhang and J.~F. Fisac, ``Safe occlusion-aware autonomous driving via game-theoretic active perception,'' in \emph{17th Robotics: Science and Systems (RSS)}, 2021.

\bibitem{game_racing}
A.~Liniger and J.~Lygeros, ``A noncooperative game approach to autonomous racing,'' \emph{IEEE Transactions on Control Systems Technology}, vol.~28, no.~3, pp. 884--897, 2019.

\bibitem{mine_1}
M.~Liu, Y.~Wan, F.~Lewis, S.~Nageshrao, and D.~Filev, ``A three-level game-theoretic decision-making framework for autonomous vehicles,'' \emph{IEEE Transactions on Intelligent Transportation Systems}, 2022.

\bibitem{Victor}
V.~G. Lopez, F.~L. Lewis, M.~Liu, Y.~Wan, S.~Nageshrao, and D.~Filev, ``Game-theoretic lane-changing decision making and payoff learning for autonomous vehicles,'' \emph{IEEE Transactions on Vehicular Technology}, vol.~71, no.~4, pp. 3609--3620, 2022.

\bibitem{online_payoff}
Q.~Zhang, R.~Langari, H.~E. Tseng, D.~Filev, S.~Szwabowski, and S.~Coskun, ``A game theoretic model predictive controller with aggressiveness estimation for mandatory lane change,'' \emph{IEEE Transactions on Intelligent Vehicles}, vol.~5, no.~1, pp. 75--89, 2019.

\bibitem{game_leader}
N.~Li, Y.~Yao, I.~Kolmanovsky, E.~Atkins, and A.~R. Girard, ``Game-theoretic modeling of multi-vehicle interactions at uncontrolled intersections,'' \emph{IEEE Transactions on Intelligent Transportation Systems}, vol.~23, no.~2, pp. 1428--1442, 2020.

\bibitem{human_reasoning_1}
P.~Hang, C.~Lv, Y.~Xing, C.~Huang, and Z.~Hu, ``Human-like decision making for autonomous driving: A noncooperative game theoretic approach,'' \emph{IEEE Transactions on Intelligent Transportation Systems}, vol.~22, no.~4, pp. 2076--2087, 2020.

\bibitem{pcpg_paper}
M.~Liu, H.~E. Tseng, D.~Filev, A.~Girard, and I.~Kolmanovsky, ``Safe and human-like autonomous driving: A predictor–corrector potential game approach,'' \emph{IEEE Transactions on Control Systems Technology}, pp. 1--15, 2023.

\bibitem{BR}
S.~Durand and B.~Gaujal, ``Complexity and optimality of the best response algorithm in random potential games,'' in \emph{International Symposium on Algorithmic Game Theory}.\hskip 1em plus 0.5em minus 0.4em\relax Springer, 2016, pp. 40--51.

\bibitem{deadlock}
R.~Mandiau, A.~Champion, J.-M. Auberlet, S.~Espi{\'e}, and C.~Kolski, ``Behaviour based on decision matrices for a coordination between agents in a urban traffic simulation,'' \emph{Applied Intelligence}, vol.~28, no.~2, pp. 121--138, 2008.

\bibitem{suzhou}
Q.~Dai, X.~Xu, W.~Guo, S.~Huang, and D.~Filev, ``Towards a systematic computational framework for modeling multi-agent decision-making at micro level for smart vehicles in a smart world,'' \emph{Robotics and Autonomous Systems}, vol. 144, p. 103859, 2021.

\bibitem{closedloop1}
W.~Schwarting, A.~Pierson, J.~Alonso-Mora, S.~Karaman, and D.~Rus, ``Social behavior for autonomous vehicles,'' \emph{Proceedings of the National Academy of Sciences}, vol. 116, no.~50, pp. 24\,972--24\,978, 2019.

\bibitem{validation_1}
N.~Ratliff, B.~Ziebart, K.~Peterson, J.~A. Bagnell, M.~Hebert, A.~K. Dey, and S.~Srinivasa, ``Inverse optimal heuristic control for imitation learning,'' in \emph{Artificial Intelligence and Statistics}, 2009, pp. 424--431.

\bibitem{validation_2}
B.~D. Ziebart, A.~L. Maas, J.~A. Bagnell, A.~K. Dey \emph{et~al.}, ``Maximum entropy inverse reinforcement learning.'' in \emph{Proceedings of the {AAAI} Conference}, vol.~8, 2008, pp. 1433--1438.

\bibitem{validation_3}
B.~D. Ziebart, A.~L. Maas, A.~K. Dey, and J.~A. Bagnell, ``Navigate like a cabbie: Probabilistic reasoning from observed context-aware behavior,'' in \emph{Proceedings of the 10th international conference on Ubiquitous computing}, 2008, pp. 322--331.

\bibitem{validation_4}
O.~Siebinga, A.~Zgonnikov, and D.~Abbink, ``A human factors approach to validating driver models for interaction-aware automated vehicles,'' \emph{ACM Transactions on Human-Robot Interaction (THRI)}, vol.~11, no.~4, pp. 1--21, 2022.

\bibitem{payoff1}
Q.~Dai, D.~Shen, J.~Wang, S.~Huang, and D.~Filev, ``Calibration of human driving behavior and preference using vehicle trajectory data,'' \emph{Transportation Research Part C: Emerging Technologies}, vol. 145, p. 103916, 2022.

\bibitem{sov}
V.-A. Le and A.~A. Malikopoulos, ``A cooperative optimal control framework for connected and automated vehicles in mixed traffic using social value orientation,'' \emph{arXiv preprint arXiv:2203.17106}, 2022.

\bibitem{game_decomposition}
O.~Candogan, I.~Menache, A.~Ozdaglar, and P.~A. Parrilo, ``Flows and decompositions of games: Harmonic and potential games,'' \emph{Mathematics of Operations Research}, vol.~36, no.~3, pp. 474--503, 2011.

\bibitem{game_decomposition_2}
D.~Cheng, T.~Liu, K.~Zhang, and H.~Qi, ``On decomposed subspaces of finite games,'' \emph{IEEE Transactions on Automatic Control}, vol.~61, no.~11, pp. 3651--3656, 2016.

\bibitem{game_book}
Y.~Shoham and K.~Leyton-Brown, \emph{Multiagent systems: Algorithmic, game-theoretic, and logical foundations}.\hskip 1em plus 0.5em minus 0.4em\relax Cambridge University Press, 2008.

\bibitem{potential_book}
Q.~D. L{\~a}, Y.~H. Chew, and B.-H. Soong, \emph{Potential Game Theory}.\hskip 1em plus 0.5em minus 0.4em\relax Springer, 2016.

\bibitem{potential_projection}
O.~Candogan, A.~Ozdaglar, and P.~A. Parrilo, ``A projection framework for near-potential games,'' in \emph{49th IEEE Conference on Decision and Control (CDC)}, 2010, pp. 244--249.

\bibitem{nonstrategic}
Y.~Wang, T.~Liu, and D.~Cheng, ``From weighted potential game to weighted harmonic game,'' \emph{IET Control Theory \& Applications}, vol.~11, no.~13, pp. 2161--2169, 2017.

\bibitem{add_dynamics_RL}
C.-J. Hoel, K.~Driggs-Campbell, K.~Wolff, L.~Laine, and M.~J. Kochenderfer, ``Combining planning and deep reinforcement learning in tactical decision making for autonomous driving,'' \emph{IEEE Transactions on Intelligent Vehicles}, vol.~5, no.~2, pp. 294--305, 2019.

\bibitem{add_dynamics_nan}
N.~Li, H.~Chen, I.~Kolmanovsky, and A.~Girard, ``An explicit decision tree approach for automated driving,'' in \emph{Dynamic Systems and Control Conference}, vol. 58271, 2017, p. V001T45A003.

\bibitem{add_dynamics_RTD}
S.~Kousik, S.~Vaskov, M.~Johnson-Roberson, and R.~Vasudevan, ``Safe trajectory synthesis for autonomous driving in unforeseen environments,'' in \emph{Dynamic Systems and Control Conference}, vol. 58271, 2017, p. V001T44A005.

\bibitem{bicycle}
J.~Kong, M.~Pfeiffer, G.~Schildbach, and F.~Borrelli, ``Kinematic and dynamic vehicle models for autonomous driving control design,'' in \emph{Proceedings of IEEE Intelligent Vehicles Symposium (IV)}, 2015, pp. 1094--1099.

\bibitem{robustness}
L.~Arditti, G.~Como, F.~Fagnani, and M.~Vanelli, ``Robustness of {N}ash equilibria in network games,'' \emph{arXiv preprint arXiv:2004.12869}, 2020.

\bibitem{bicycle_2}
J.~Kong, M.~Pfeiffer, G.~Schildbach, and F.~Borrelli, ``Kinematic and dynamic vehicle models for autonomous driving control design,'' in \emph{IEEE Intelligent Vehicles Symposium (IV)}, 2015, pp. 1094--1099.

\bibitem{bicycle_validation}
P.~Polack, F.~Altch{\'e}, B.~d'Andr{\'e}a Novel, and A.~de~La~Fortelle, ``The kinematic bicycle model: A consistent model for planning feasible trajectories for autonomous vehicles?'' in \emph{IEEE Intelligent Vehicles Symposium (IV)}, 2017, pp. 812--818.

\bibitem{add_T}
L.~Yang, X.~Li, W.~Guan, H.~M. Zhang, and L.~Fan, ``Effect of traffic density on drivers’ lane change and overtaking maneuvers in freeway situation—a driving simulator--based study,'' \emph{Traffic Injury Prevention}, vol.~19, no.~6, pp. 594--600, 2018.

\end{thebibliography}
\bibliographystyle{IEEEtran}

\begin{IEEEbiography}[{\includegraphics[width=1in,height=1.25in,clip,keepaspectratio]{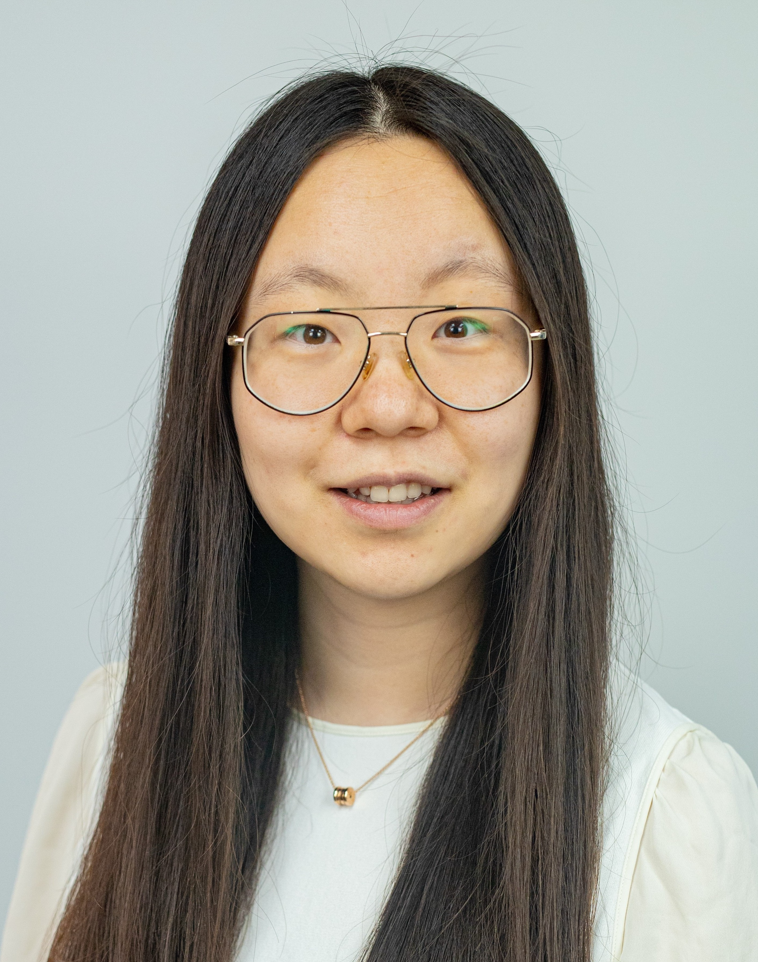}}]{Mushuang Liu} is currently an Assistant Professor in the Department of Mechanical and Aerospace Engineering at the University of Missouri, Columbia, MO. She worked as a postdoc in the Department of Aerospace Engineering at the University of Michigan, Ann Arbor, MI, 2021-2022. She received her Ph.D degree from the University of Texas at Arlington in 2020 and her B.S. degree from the University of Electronic Science and Technology of China in 2016. Her research interests include intelligent control and decision-making for multi-agent systems using techniques from control theory, game theory, and machine learning. The applications of her research include automotive systems, aerospace systems, and robotics.
\end{IEEEbiography}

\begin{IEEEbiography}[{\includegraphics[width=1in,height=1.25in,clip,keepaspectratio]{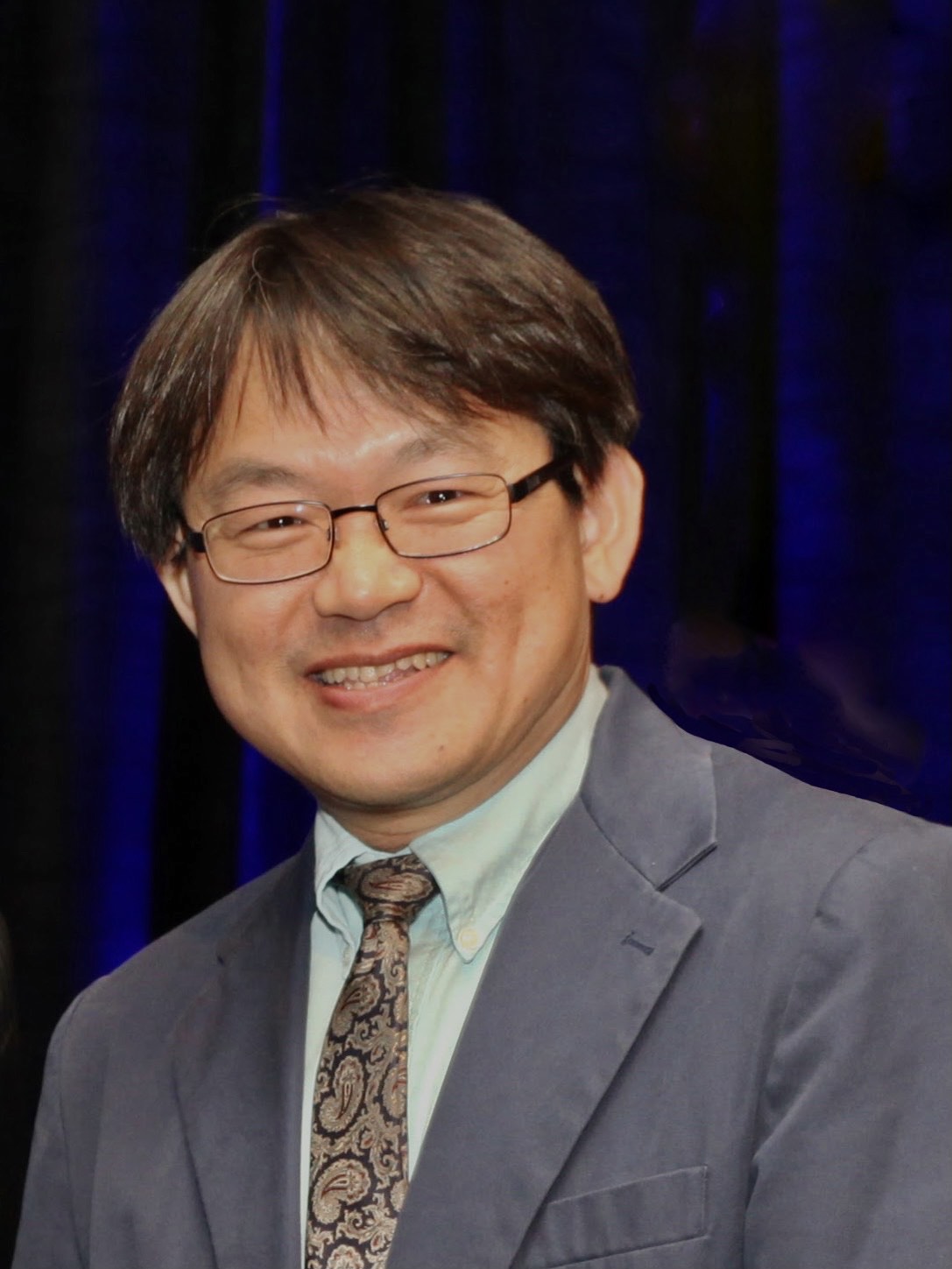}}]{H. Eric Tseng}  received the B.S. degree from the National Taiwan University, Taipei, Taiwan, in 1986, and the M.S. and Ph.D. degrees in mechanical engineering from the University of California at Berkeley, Berkeley, in 1991 and 1994, respectively. In 1994, he joined Ford Motor Company. At Ford, he is currently a Senior Technical Leader of Controls and Automated Systems in Research and Advanced Engineering. Many of his contributed technologies led to production vehicles implementation. His technical achievements have been recognized internally seven times with Ford’s highest technical award—the Henry Ford Technology Award, as well as externally by the American Automatic Control Council with Control Engineering Practice Award in 2013. He has over 100 U.S. patents
and over 120 publications. He is an NAE Member.
\end{IEEEbiography}

\begin{IEEEbiography}[{\includegraphics[width=1in,height=1.5in,clip,keepaspectratio]{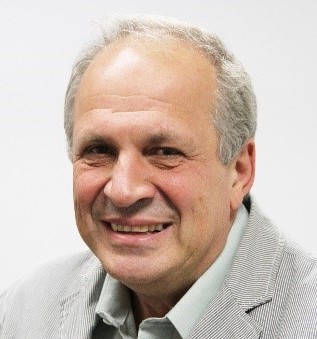}}]{Dimitar Filev} (Fellow, IEEE) is Senior Henry Ford Technical Fellow in Control and AI with Research $\&$ Advanced Engineering – Ford Motor Company. His research is in computational intelligence, AI and intelligent control, and their applications to autonomous driving, vehicle systems, and automotive engineering.  He holds over 100 granted US patents and has been awarded with the IEEE SMCS 2008 Norbert Wiener Award and the 2015 Computational Intelligence Pioneer’s Award. Dr. Filev is a Fellow of the IEEE and a member of the National Academy of Engineering. He was President of the IEEE Systems, Man, and Cybernetics Society (2016-2017).
\end{IEEEbiography}

\begin{IEEEbiography}[{\includegraphics[width=1in,height=1.5in,clip,keepaspectratio]{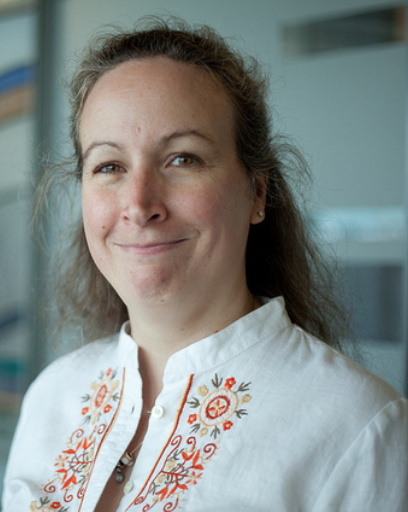}}]{Anouck Girard} received
the Ph.D. degree in ocean engineering from the
University of California at Berkeley, Berkeley, CA, USA, in 2002. She has been with the University of Michigan, Ann Arbor, MI, USA, since 2006, where she is currently a Professor of aerospace engineering. She has coauthored the book Fundamentals of Aerospace Navigation and Guidance (Cambridge University Press, 2014). Her current research interests include vehicle dynamics and control systems. She was a recipient of the Silver Shaft Teaching Award from the University of Michigan and the Best Student Paper Award from the American Society of Mechanical Engineers.
\end{IEEEbiography}

\begin{IEEEbiography}[{\includegraphics[width=1in,height=1.25in,clip,keepaspectratio]{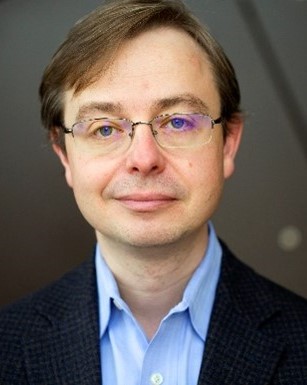}}]{Ilya Kolmanovsky} (Fellow, IEEE) received the Ph.D. degree in aerospace engineering from the University of Michigan, Ann Arbor, MI, USA, in 1995. Prior to joining the University of Michigan as a Faculty Member in 2010, he was with Ford Research and Advanced Engineering, Dearborn, MI, for close to 15 years. He is currently a Pierre T. Kabamba Collegiate Professor in the Department of Aerospace Engineering at the University of Michigan. His research interests include control theory for systems with state and control constraints, and control applications to aerospace and automotive systems. He is a Fellow of IFAC and NAI and a Senior Editor of IEEE TRANSACTIONS ON CONTROL SYSTEMS TECHNOLOGY.
\end{IEEEbiography}
\end{document}